\documentclass[sigconf]{acmart}
\usepackage{color}
\usepackage{soul}
\usepackage{graphicx}
\usepackage[ruled,vlined]{algorithm2e}
\usepackage[noend]{algorithmic}

\usepackage{chngcntr}
\usepackage{epstopdf}

\usepackage{caption}
\usepackage[labelformat=simple]{subcaption}

\makeatletter
\newcounter{module}
\newenvironment{module}[1][htb]
  {
   \let\c@algocf\c@module
   \begin{algorithm}[#1]%
  }{\end{algorithm}}
\makeatother

\setcopyright{none}
\renewcommand\footnotetextcopyrightpermission[1]{} 
\newcommand{\pvec}[1]{\vec{#1}\mkern2mu\vphantom{#1}}

\newcommand{\D}{\mathrm{\textbf{D}}}
\newcommand{\tet}{\vec{\theta}}
\newcommand{\A}{\mathrm{\textbf{A}}}
\newcommand{\rank}{\mathrm{rank}}
\newcommand{\nul}{\mathrm{Null}}
\newcommand{\supp}{\mathrm{supp}}

\newcommand{\intt}{\mathrm{int}}
\newcommand{\W}{\mathrm{\textbf{W}}}
\newcommand{\I}{\mathrm{\textbf{I}}}
\newcommand{\diag}{\mathrm{diag}}
\newcommand{\cl}{\mathrm{cl}}

\begin{document}
\title{REACT to Cyber Attacks on Power Grids}
\author{Saleh Soltan}
\affiliation{Electrical Engineering, Princeton University, New Jersey, NJ}
\email{ssoltan@princeton.edu}
\author{Mihalis Yannakakis}
\affiliation{Computer Science, Columbia University, New York, NY}
\email{mihalis@cs.columbia.edu}
\author{Gil Zussman}
\affiliation{Electrical Engineering, Columbia University, New York, NY}
\email{gil@ee.columbia.edu}

\setlength{\textfloatsep}{4 pt}
\begin{abstract}
Motivated by the recent cyber attack on the Ukrainian power grid, we study cyber attacks on power grids that affect both the physical infrastructure and the data at the control center. In particular, we assume that an adversary attacks an area by: (i) remotely disconnecting some lines within the attacked area, and (ii) modifying the information received from the attacked area to mask the line failures and hide the attacked area from the control center. For the latter, we consider two types of attacks: (i) \emph{data distortion:} which distorts the data by adding powerful noise to the actual data, and (ii) \emph{data replay:} which replays a locally consistent old data instead of the actual data. We use the DC power flow model and prove that the problem of finding the set of line failures given the phase angles of the nodes outside of the attacked area is strongly NP-hard, even when the attacked area is known. However, we introduce the polynomial time REcurrent Attack Containment and deTection (REACT) Algorithm to approximately detect the attacked area and line failures after a cyber attack. We numerically show that it performs very well in detecting the attacked area, and detecting single, double, and triple line failures in small and large attacked areas.

\end{abstract}

\settopmatter{printfolios=true} 
\maketitle
\section{Introduction}\label{sec:intro}

Due to their complexity and magnitude, modern infrastructure networks need to be monitored and controlled using computer systems. These computer systems are vulnerable to cyber attacks~\cite{radics}.  One of the most important infrastructure networks that is vulnerable to cyber attacks is the power grid which is monitored and controlled by the Supervisory  Control And Data Acquisition (SCADA) system.

In a recent cyber attack on the Ukrainian power grid~\cite{UkraineBlackout}, the attackers stole credentials for accessing the SCADA system and used it to cause a large scale blackout affecting hundred thousands of people. In particular, they simultaneously operated several of the circuit breakers in the grid and jammed the phone lines to keep the system operators unaware of the situation~\cite{UkraineBlackout}.

\begin{figure}[t]
\centering
\vspace*{0.2cm}
\includegraphics[scale=0.3]{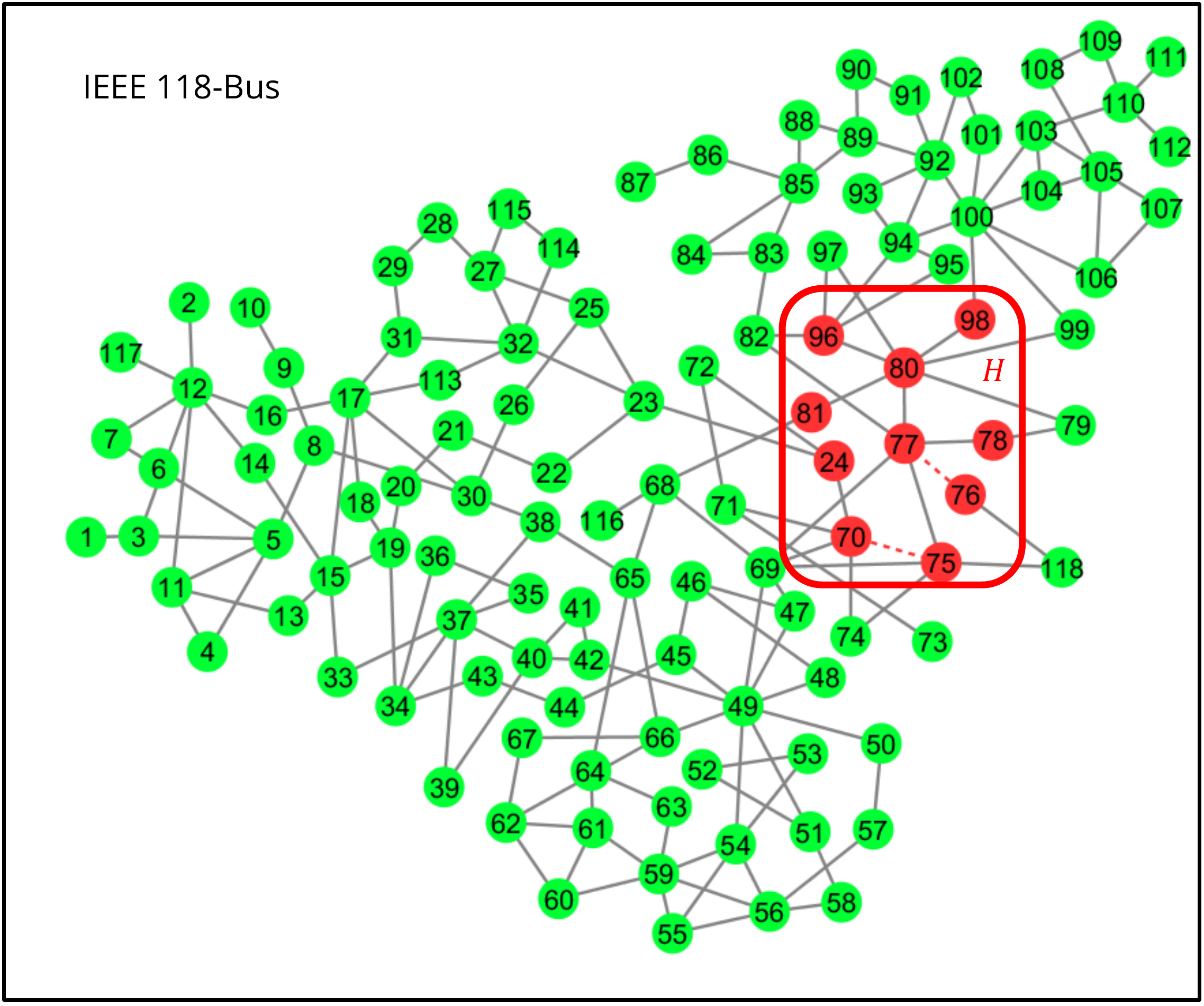}
\vspace*{-0.2cm}
\caption{The attack model.
An adversary attacks an area $H$ which is unknown to the control center (represented by red nodes) by disconnecting some lines within the attacked area (shown by red dashed lines) and modifying the information received from the attacked area to mask the line failures and hide the attacked area from the control center.}
\label{fig:attack_model}
\end{figure}

 Motivated by the Ukraine event, in this paper, we deploy the DC power flow model and study a model of a cyber attack on the power grid that affects both the physical infrastructure and the data at the control center. As illustrated in Fig.~\ref{fig:attack_model}, we assume that an adversary attacks an area by: (i) disconnecting some lines within the attacked area (by remotely activating the circuit breakers), and (ii) modifying the information (phase angles of the nodes and status of the lines) received from the attacked area  to mask the line failures and hide the attacked area from the control center. For the latter, we consider two types of attacks: (i) \emph{data distortion:} which distorts the data by adding powerful noise to the data received from the attacked area, and (ii) \emph{data replay:} which replays a locally consistent old data instead of the actual data. 

We prove that the problem of finding the set of line failures given the phase angles of the nodes outside of the attacked area is strongly NP-hard, \emph{even when the attacked area is known}. Hence, one cannot expect to develop a polynomial time algorithm that can exactly detect the attacked area and recover the information for all possible attack scenarios. However, we introduce the polynomial time REcurrent Attack Containment and deTection (REACT) Algorithm and numerically show that it performs very well in reasonable scenarios.

In particular, we first introduce the ATtacked Area Containment (ATAC) Module for approximately detecting the attacked area using graph theory and the algebraic properties of the DC power flow equations. We show that the ATAC Module can always provide an area containing the attacked area after a data distortion or a data replay attack. We further provide tools to improve the accuracy of the approximated attacked areas obtained by the ATAC Module under different data attack types.

Then, we introduce the randomized LIne Failures Detection (LIFD) Module to detect the line failures and recover the phase angles inside the detected attacked area. The LIFD Module builds upon the methods first introduced in~\cite{SYZ2015}, to detect line failures using Linear Programming (LP) in more general cases. In particular, we prove that in some cases that the methods in~\cite{SYZ2015} fail to detect line failures, the LIFD Module can successfully detect line failures in expected polynomial running time.  


Finally, the REACT Algorithm combines the ATAC and LIFD Modules to provide a comprehensive algorithm for attacked area detection and information recovery following a cyber attack. We evaluate the performance of the REACT Algorithm by considering two attacked areas, one with 15 nodes and the other one with 31 nodes within the IEEE 300-bus system~\cite{IEEEtestcase}. We show that when the attacked area is small, the REACT Algorithm performs equally well after the data distortion and the data replay attacks. In particular, it can exactly detect the attacked area in all the cases, and accurately detect single, double, and triple line failures  within the attacked area in more than 80\% of the cases.

When the attacked area is large, however, the REACT Algorithm's performance is different after the data distortion and the data replay attacks. It still performs very well in detecting the attacked area after a data distortion attack and accurately detects line failures after single, double, and triple line failures in more than 60\% of the cases. However, it may face difficulties providing an accurate approximation of the attacked area after a replay attack. Despite these difficulties in approximating the attacked area, it accurately detects single and double line failures in around 98\% and 60\% of the cases, respectively.

The main contributions of this paper are two folds: (i) analyzing the computational complexity of the attacked area detection and information recovery problem after a cyber attack on the grid, and (ii) introducing a polynomial time algorithm to address this problem and numerically evaluating its performance. To the best of our knowledge, \emph{this is the first attempt to develop recovery algorithms for the attacks in which the data from the area is modified and therefore the attacked area in unknown.} 
\section{Related Work}\label{sec:related}

The vulnerability of general networks to attacks was thoroughly studied in the past (e.g., \cite{albert2000error,phillips1993network,Kleinberg2004NFD} and references therein). In particular,
\cite{Ciavarella2017Progressive,tootaghaj2017network} studied a problem similar to the one studied in this paper (failure detection from partial observations) in the context of communication networks.


Vulnerability of power grids to failures and attacks was extensively studied~\cite{nesti2016reliability,pinar_power,kim2016analyzing,liu2014distributed,Bern2012ACM,Dobson,bienstock2016electrical,SYZ2015}. In particular false data injection attacks on power grids and anomaly detection were studied using the DC power flows in~\cite{kim2013topology,liu2011false,dan2010stealth,vukovic2011network,li2015quickest,kim2015subspace}. These studies focused on the observability of the failures and attacks in the grid. 


The problem studied in this paper is related to the problem of line failures detection using phase angle measurements~\cite{tate2008line,tate2009double,garcia2016line,zhu2012sparse,SYZ2015}. Up to two line failures detection, under the DC power flow model, was studied in~\cite{tate2008line,tate2009double}. Since the methods developed in~\cite{tate2008line,tate2009double} are greedy-based that need to search the entire failure space, their the running time grows exponentially as the number of failures increases. Hence, these methods cannot be generalized to detect higher order failures.  Similar greedy approaches with likelihood detection functions  were studied in~\cite{manousakis2012taxonomy,Khandeparkar2014Eff,zhao2012pmu,zhao2014identification,zhu2014phasor} to address the PMU placement problem under the DC power flow model. 

 The problem of line failures detection in an internal system using the information from an external system was also studied in~\cite{zhu2012sparse} based on the DC power flow model. The proposed algorithm works for only one and two line failures, since it depends on the sparsity of line failures. 
In a recent work~\cite{garcia2016line}, a linear multinomial regression model was proposed as a classifier for a single line failure detection using transient voltage phase angles data. Due to the time complexity of the learning process for multiple line failures, this method is impractical for detecting higher order failures. 

In~\cite{SYZ2015}, attack scenarios similar to the one in this paper was studied. However, \cite{SYZ2015} only focused on the attacks that blocked the information from the attacked area, and therefore, the attacked area was detectable simply by checking the missing data. In this work, we build upon the results of~\cite{SYZ2015} to detect line failures in more general data attack cases than the ones considered in~\cite{SYZ2015}. In a recent work~\cite{soltan2017power}, the methods provided in~\cite{SYZ2015} were extended to function under the AC power flow model.

Finally, in a recent series of works, the vulnerability of power grids to undetectable cyber-physical attacks is studied~\cite{li2016bilevel,deng2017ccpa,zhang2016physical} using the DC power flows. These studies are mainly focused on designing attacks that affect the entire grid and therefore may be impossible to detect. 


\section{Model and Definitions}\label{sec:Model}
\subsection{DC Power Flow Model}
\label{ssec:flow-model}
We adopt the linearized DC power flow model, which is widely used as an approximation for the non-linear AC power flow model~\cite{bergenvittal}. The notation is summarized in Table~\ref{tb:notatition}.
In particular, we represent the power grid by a connected undirected graph $G=(V,E)$ where $V=\{1,2,\dots,n\}$ and $E=\{e_1,\dots,e_m\}$ are the set of nodes and edges corresponding to the buses and transmission lines, respectively. Each edge $e_i$ is a set of two nodes $e_i=\{u,v\}$.
$p_v$ is the active power \emph{supply} ($p_v>0$) or \emph{demand} ($p_v<0$) at node $v\in V$ (for a \emph{neutral node} $p_v=0$).
We assume \emph{pure reactive} lines, implying that each edge $\{u,v\} \in E$ is characterized by its \emph{reactance} $r_{uv}=r_{vu}$.

Given the power supply/demand vector $\vec{p}\in \mathbb{R}^{|V|\times1}$ and the reactance values, a \emph{power flow} is a solution $\mathrm{\textbf{P}}\in \mathbb{R}^{|V|\times|V|}$ and $ \tet\in\mathbb{R}^{|V|\times 1}$ of:
\begin{eqnarray}
\label{eqn:flow1}&&\sum_{v \in N(u)}p_{uv} = p_u, \ \forall~ u \in V \\
\label{eqn:flow2}&&\theta_u - \theta_v - r_{uv}p_{uv} = 0, \ \forall~ \{u,v\} \in E
\end{eqnarray}
where $N(u)$ is the set of neighbors of node $u$, $p_{uv}$ is the power flow from node $u$ to node $v$, and $\theta_u$ is the phase angle of node $u$.
Eq.~(\ref{eqn:flow1}) guarantees (classical) flow conservation and (\ref{eqn:flow2}) captures the dependency of the flow on the reactance values and phase angles. Additionally, (\ref{eqn:flow2}) implies that $p_{uv}=-p_{vu}$.
When the total supply equals the total demand in each connected component of $G$, (\ref{eqn:flow1})-(\ref{eqn:flow2}) has a unique solution $\textbf{P}$ and $\tet$ up to a shift (since shifting all $\theta_u$s by equal amounts does not violate (\ref{eqn:flow2})). 
Eqs.(\ref{eqn:flow1})-(\ref{eqn:flow2}) are equivalent to the following matrix equation:
\begin{equation}\label{eqn:flowmatrix}
\textbf{A}\vec{\theta}=\vec{p}
\end{equation}
where $\textbf{A}\in \mathbb{R}^{|V|\times |V|}$ is the \textit{admittance matrix} of $G$,\footnote{The matrix $\textbf{A}$ can also be considered as the \emph{weighted Laplacian matrix} of the graph.} defined as:
\begin{equation*}\label{def: plus}
a_{uv}=
\begin{cases}
0&\text{if}~ u\neq v~\text{and}~\{u,v\}\notin E,\\
-1/r_{uv}&\text{if}~u\neq v~\text{and}~\{u,v\}\in E,\\
-\sum_{w\in N(u)} a_{uw}&\text{if}~u=v.
\end{cases}
\end{equation*}
Note that in power grids nodes can be connected by multiple edges,  and therefore, if there are $k$ multiple lines between nodes $u$ and $v$, $a_{uv}=-\sum_{i=1}^{k}1/r_{uv_i}$. Once $\vec{\theta}$ is computed, the flows, $p_{uv}$, can be obtained from~(\ref{eqn:flow2}). \\
\textbf{Notation.} Throughout this paper we use bold uppercase characters to denote matrices (e.g., $\textbf{A}$), italic uppercase characters to denote sets (e.g., $V$), and italic lowercase characters and overline arrow to denote column vectors (e.g., $\vec{\theta}$). For a matrix $\textbf{Q}$, $\textbf{Q}_i$ denotes its $i^{th}$ row, and $q_{ij}$ denotes its $(i,j)^{\text{th}}$ entry. For a column vector $\vec{y}$, $\pvec{y}^T$ denote its transpose, $y_i$ denotes its $i^{\text{th}}$ entry, $\|\vec{y}\|_1:=\sum_{i=1}^n |y_i|$ is its $l_1$-norm, and $\supp(\vec{y}):=\{i|y_i\neq0\}$ is its support.

\begin{table}[t]
\caption{Summary of notation.}
\vspace*{-0.2cm}
\centering
\begin{tabular}{|c|l|}
\hline
\textbf{Notation}&\textbf{Description}\\
\hline
$G=(V,E)$& The graph representing the power grid\\
\hline
$\A$ & Admittance matrix of $G$\\
\hline
$\tet$ & Vector of the phase angles of the nodes in $G$\\
\hline
$\vec{p}$ & Vector of power supply/demand values\\
\hline
$H$ & A subgraph of $G$ representing the attacked area\\
\hline
$F$ & Set of failed edges due to an attack\\
\hline
$\D$ & Incidence matrix of $G$\\
\hline
$N(i)$& Set of neighbors of node $i$\\
\hline
$N(S)$& Set of neighbors of subgraph $S$\\
\hline
$\intt(S)$ &Interior of the subgraph $S$\\
\hline
$\partial(S)$& Boundary of the subgraph $S$\\
\hline
$\cl(S)$& Closure of the subgraph $S$\\
\hline
$\bigcirc'$& The actual value of $\bigcirc$ after an attack\\
\hline
$\bigcirc^{\star}$& The observed value of $\bigcirc$ after an attack\\
\hline
$\overline{\bigcirc}$& The complement of $\bigcirc$\\
\hline
\end{tabular}\label{tb:notatition}
\end{table}
\subsection{Attack Model}\label{subsec:attack}
We study a cyber attack on the power grid that affects both its physical infrastructure and the data at its control center. We assume that an adversary attacks an area by: (i) disconnecting some lines within the attacked area (by remotely activating the circuit breakers), and (ii) modifying the information (phase angle of the nodes and status of the lines) received from the attacked area  to mask the line failures and hide the attacked area from the control center. 
We assume that supply/demand values do not change after the attack and disconnecting lines within the attacked area does not make $G$ disconnected. However, the developed methods in this paper can also be used when these conditions do not hold, if the control center is aware of the changes in the supply/demand values after the attack and in the case of the grid separation.

Fig.~\ref{fig:attack_model} shows an example of such an attack on the area represented by $H=(V_H, E_H)$. Due to the attack, some edges are disconnected (we refer to these edges as \emph{failed lines})  and the phase angles and the status of the lines within the \emph{attacked area} are modified. We denote the set of failed lines in area $H$ by $F\subseteq E_H$. Upon failure, the failed lines are removed from the graph and the flows are redistributed according to (\ref{eqn:flow1})-(\ref{eqn:flow2}). The objective is to detect the attacked area and the failed lines after the attack using the observed phase angles.

The vectors of phase angle of the nodes in $H$ and in its complement $\bar{H} = G\backslash H$ are denoted by $\tet_H$ and $\tet_{\bar{H}}$, respectively. We use the prime symbol $(')$ to denote the actual values after an attack. For instance, $G'$, $\A'$, and $\vec{\theta}'$ are used to represent the graph, the admittance matrix of the graph, and the actual phase angles after the attack. Based on our assumptions $\vec{p}=\A\tet =\A'\tet'=\pvec{p}'$. 

We also use $\tet^{\star}$ to denote the observed phase angles after the attack. According to the attack model $\tet_H^{\star}$ is modified and is not necessarily equal to $\tet_H'$.
We assume that the attacker performs any of the following two types of data attacks:
\vspace*{-0.1cm}
\begin{enumerate}
\item \textbf{Data distortion:} We assume $\tet_H^{\star}=\tet_H'+\vec{z}$ for a random vector $\vec{z}$ with an arbitrary distribution with no positive probability mass in any proper linear subspace (e.g., multivariate Gaussian distribution).
\vspace*{-0.1cm}
\item \textbf{Data replay:} We assume $\tet_H^{\star}=\tet_H''$ such that $\tet''$ satisfies $\A\tet'' =\pvec{p}''$ for an arbitrary power supply/demand vector $\pvec{p}''$ such that $\pvec{p}_H'' =\pvec{p}_H$. We assume that $\pvec{p}_{\bar{H}}''$ is selected generally enough and is only known to the attacker. $\pvec{p}''$ can be considered as the vector of supply/demand values from previous hours or days.
\end{enumerate}
\vspace*{-0.1cm}
 Notice that adversarial modification of the reported phase angles in $H$ is not in the scope of this paper and is an interesting problem on its own.


\textbf{Notation.} 
Without loss of generality we assume that the indices are such that $V_H=\{1,2,\dots,|V_H|\}$ and $E_H=\{e_1,e_2,\dots,e_{|E_H|}\}$. 
If $X,Y$ are two subgraphs of $G$, $\A_{X|Y}$ and $\A_{V_X|V_Y}$ both denote the submatrix of the admittance matrix of $G$ with rows from $V_X$ and columns from $V_Y$. For instance, $\A$ can be written in any of the following forms,
\begin{equation*}
\A=
\begin{bmatrix}
\A_{H|H}&\A_{H|\bar{H}}\\
\A_{\bar{H}|H}&\A_{\bar{H}|\bar{H}}
\end{bmatrix}
, \A=
\begin{bmatrix}
\A_{G|H}&\A_{G|\bar{H}}
\end{bmatrix}
, \A=
\begin{bmatrix}
\A_{H|G}\\
\A_{\bar{H}|G}
\end{bmatrix}.
\end{equation*}


\subsection{Graph Theoretical Terms}\label{subsec:graph}
In this paper, we use some graph theoretical terms most of which are borrowed from~\cite{bondy2008graph}.

\noindent\textbf{Subgraphs:} Let $X$ be a subset of the nodes of a graph $G$. $G[X]$ denotes the subgraph of $G$ \emph{induced} by $X$. 
We denote the complement of a set $X$ by $\bar{X}=V\backslash X$.

The \emph{neighbors}, \emph{interior}, \emph{boundary}, and \emph{closure} of a subgraph $S$ are defined and denoted by $N(S):=\{i\in V\backslash V_S|\exists j\in V_S: i\in N(j) \}$, $\intt(S):=\{i\in V_S|N(i)\subseteq V_S\}$, $\partial(S):=\{i\in V_S|N(i)\cap V_{\bar{S}}\neq \emptyset\}$, and $\cl(S):=V_S\cup N(S)$, respectively.



%

\noindent\textbf{Incidence Matrix:} Assign arbitrary directions to the edges of $G$. 
The (node-edge) \emph{incidence matrix} of $G$ is denoted by $\D\in\{-1,0,1\}^{|V|\times|E|}$ and is defined as follows,
\begin{equation*}
d_{ij}=
\begin{cases}
0&\text{if}~e_j~\text{is not incident to node}~i,\\
1&\text{if}~e_j~\text{is coming out of node}~i,\\
-1&\text{if}~e_j~\text{is going into node}~i.\\
\end{cases}
\end{equation*}
When we use the incidence matrix, we assume an arbitrary orientation for the edges unless we mention an specific orientation. $\D_H\in\{-1,0,1\}^{|V_H|\times|E_H|}$ is the submatrix of $\D$ with rows from $V_H$ and columns from $E_H$. 
\section{Hardness}\label{sec:hardness}
Using the notation provided in the previous section, the problem considered in this paper can be stated as follows: Given $\A, \tet,$ and $\tet^{\star}$, detect the attacked area $H$ and the set of line failures $F$. In this section, we study the computational complexity of this and related problems. To study the computational complexity of this problem, we consider a more general case of $\tet^{\star}_H$ without any assumptions on the type of the data attack.

First, we prove that the problem of finding the set of line failures ($F$) solely based on the given the phase angles of the nodes before ($\tet$) and after the attack ($\tet'$) is NP-hard. We prove this by reduction from the \emph{3-partition problem}.
\begin{definition}\label{def:3-partition}
Given a set $S=\{s_1,s_2,\dots,s_{3k}\}$ of $3k$ elements and a bound $B$, such that $\sum_{i=1}^{3k} s_i=k B$ and for $1\leq i\leq 3k$, $B/4<s_i<B/2$, the \emph{3-partition problem} is the problem of whether $S$ can be partitioned into $k$ disjoint sets $S_1,\dots,S_k$ such that for $1\leq i\leq k$, $\sum_{s_j\in S_i} s_j = B$ (note that each $S_i$ must therefore contain exactly 3 elements from $S$).
\end{definition}
\begin{lemma}[Garey and Johnson\cite{garey1979computers}]\label{lem:3-partition}
The 3-partition problem is strongly NP-complete.
\end{lemma}
\begin{lemma}\label{lem:hardness_no_cyber}
Given $\A, \tet,$ and $\tet'$, it is strongly NP-hard to determine if there exists a set of line failures $F$ such that $\A'\tet'=\A\tet$.
\end{lemma}
\begin{proof}
We reduce the 3-partition problem to this problem. Assume $S$ is a given set as described in Def.~\ref{def:3-partition}, we form a bipartite graph $G=(V,E)$ such that
$V = X\cup Y$, $E=\{\{x,y\}|x\in X, y\in Y\}$, $X=\{1,\dots,k\}$, and $Y=\{k+1,\dots,4k\}$. For all edges in $G$,  we set the reactance values equal to 1. For each $i\in X$, we set $p_{i}=B$ and for each $j\in Y$ we set $p_{j}=-s_{j-k}$. 
Define the vector of phase angles $\tet$ as follows:
\begin{equation*}
\theta_i=
\begin{cases}
0&i\leq k\\
-s_{i-k}/k& i>k.
\end{cases}
\end{equation*}
If $\A$ is the admittance matrix of $G$, it is easy to check that $\A\tet=\vec{p}$. Now define $\tet'$ as follows:
\begin{equation*}
\theta_i'=
\begin{cases}
0&i\leq k\\
-s_{i-k}& i>k.
\end{cases}
\end{equation*}
We prove that there exist a set of line failures $F$ such that $\A'\tet'=\vec{p}$ if, and only if, there exists a solution to the 3-partition problem.

First, lets assume that there exist a solution to the 3-partition problem such as $S_1,\dots,S_k$. Set $E_S=\{\{i,j\}|s_{j-k}\in S_i\}$. We show that $F=E\backslash E_S$ implies $\A'\tet'=\vec{p}$. Notice that $F=E\backslash E_S$ means that $G'=(V,E_S)$. Given the $p_i$ and the reactance values, it is easy to check that the defined $\tet'$ satisfies the DC power flow equations~(\ref{eqn:flow1})-(\ref{eqn:flow2}) in $G'$. Hence, $\A'\tet'=\vec{p}$.

Now, lets assume there exist a set of line failures $F$ such that $\A'\tet'=\vec{p}$. Set $E_S=E\backslash F$ and $G'=(V,E_S)$. Given the phase angles $\tet'$, it is easy to see that for any $\{i,j\}\in E_S$, $p_{ij}=s_{j-k}$. This implies that for $j\in Y$, at most one edge in $E_S$ is incident to $j$. On the other hand, using (\ref{eqn:flow1}), for any $i\in X$, $\sum_{j\in N(i)'} s_{j-k}=B$ in which by $N(i)'$ we mean the set of neighbors of node $i$ in $G'$. Given that each node $j\in Y$ is incident to at most one edge in $E_S$, defining $S_i=\{s_{j-k}|j\in N(i)'\}$ for $1\leq i\leq k$ gives a good solution to the 3-partition problem.

Hence, determining if there exist a set of line failures $F$ is at least as hard as determining if the 3-partition problem has a solution, and therefore, it is an NP-hard problem in the strong sense.
\end{proof}
\begin{corollary}\label{cor:hardness_no_cyber}
Given $\A, \tet,$ and $\tet'$, it is strongly NP-hard to find the set of line failures $F$, even if such a set exists.
\end{corollary}
\begin{proof}
It is easy to see that if one can find a set of line failures $F$ with an algorithm, the output of that algorithm can be used here to verify the correctness and existence of such a set as well. Therefore, this problem is at least as hard as the existence problem.
\end{proof}
In Corollary~\ref{cor:hardness_no_cyber}, we proved that given the phase angle of the nodes before and after the attack, it is NP-hard to detect the set of line failures $F$. In the following lemma, we show that even if the attack area $H$ is known (since $\tet_H'$ is not given) the problem remains NP-hard.
\begin{lemma}\label{lem:hardness_H}
Given $\A, \tet, H,$ and $\tet_{\bar{H}}'$, it is strongly NP-hard to determine if there exist a set of line failures $F$ in $H$ and a vector $\tet_H'$, such that $\A'\tet'=\A\tet$.
\end{lemma}
\begin{proof}
The idea of the proof is very similar to the proof of Lemma~\ref{lem:hardness_no_cyber}. Again we reduce the 3-partition problem with a given set $S$ as described in Def.~\ref{def:3-partition} to this problem. Consider sets $X_1=\{1,\dots,k\}$, $X_2=\{k+1,\dots,2k\}$, $Y_2=\{2k+1,\dots,5k\}$, $Y_1=\{5k+1,\dots,8k\}$. We form a bipartite graph $G=(V,E)$ such that
$V = X_1\cup X_2 \cup Y_2\cup Y_1$ and $E=\{\{i,k+i\}|1\leq i\leq k\}\cup \{\{x,y\}|x\in X_2, y\in Y_2\}\cup\{\{j,j+3k\}|2k+1\leq j\leq5k\}$. Notice that the defined bipartite graph here is very similar to the one defined in the proof of Lemma~\ref{lem:hardness_no_cyber} except that here for each node in $X_2$ and $Y_2$ there exist a dummy node in $X_1$ and $Y_1$, accordingly, that is directly connected to its counterpart. We set $H=G[X_2\cup Y_2]$. It is easy to see that $H$ has exactly the same topology as the graph $G$ in the proof of Lemma~\ref{lem:hardness_no_cyber}. Again for all edges in $G$,  we set the reactance values equal to 1. For each $i\in X_2\cup Y_2$ we set $p_i=0$, for each $i\in X_1$, we set $p_{i}=B$, and for each $j\in Y_1$ we set $p_{j}=-s_{j-5k}$.
Define the vector of phase angles $\tet$ as follows:
\begin{equation*}
\theta_i=
\begin{cases}
B&1\leq i\leq k\\
0&k+1\leq i\leq 2k\\
-s_{i-2k}/k& 2k+1\leq i\leq 5k\\
-s_{i-5k}/k-s_{i-5k}& 5k+1\leq i\leq 8k
\end{cases}
\end{equation*}
If $\A$ is the admittance matrix of $G$, it is easy to check that $\A\tet=\vec{p}$. Now define $\tet'$ as follows:
\begin{equation*}
\theta_i'=
\begin{cases}
B&1\leq i\leq k\\
0&k+1\leq i\leq 2k\\
-s_{i-2k}& 2k+1\leq i\leq 5k\\
-2s_{i-5k}& 5k+1\leq i\leq 8k
\end{cases}
\end{equation*}
Now given $\tet_{\bar{H}}'$, since each node in $H$ is connected to an exactly one distinct node in $\bar{H}$, there exist a matching between the nodes in $H$ and $\bar{H}$ that covers nodes in $H$ and therefore from \cite[Corollary 2]{SYZ2015}, $\tet_H'$ will be determined uniquely. Hence, we can assume that $\tet'$ is given for all the nodes. Now we prove that there exist a set of line failures $F$ in $H$ such that $\A'\tet'=\vec{p}$ if, and only if, there exist a solution to the 3-partition problem. Given the way we build the graph $G$ and since the set of failures should be in $H$, the rest of the proof is exactly similar to the proof of Lemma~\ref{lem:hardness_no_cyber}.
\end{proof}
\begin{corollary}\label{cor:hardness_H}
Given $\A, \tet, H,$ and $\tet_{\bar{H}}'$, it is strongly NP-hard to find the set of line failures $F$ in $H$, even if such a set exists.
\end{corollary}

Finally, we prove that when the phase angles are modified ($\tet^{\star}$) and therefore $H$ is not known in advance, it is NP-hard to detect $H$ and $F$. We assume that the attacked area cannot contain more than half of the nodes, otherwise this problem might have many solutions.
\begin{lemma}\label{lem:hardness_cyber}
Given $\A, \tet,$ and $\tet^{\star}$, it is strongly NP-hard to determine if there exists a subgraph $H_0$ with $|V_{H_0}|\leq|V|/2$, a set of line failures $F$ in $H_0$, and a vector $\tet'_{H_0}$ such that $\A\tet = \A'\left[\begin{smallmatrix}\tet'_{H_0}\\\tet^{\star}_{\bar{H}_0}\end{smallmatrix}\right]$.
\end{lemma}
\begin{proof}
Again we reduce the 3-partition partition problem to this problem. The proof is similar to the proof of Lemma~\ref{lem:hardness_H}. Given an instance of a 3-partition problem, we build a graph $G$, subgraph $H$, and supply and demand vector $\vec{p}$ exactly as in the proof of Lemma~\ref{lem:hardness_H}. Define $\tet^{\star}_{\bar{H}}=\tet'_{\bar{H}}$ as defined in Lemma~\ref{lem:hardness_H} and $\tet^{\star}_H = \vec{z}$, in which $\vec{z}$ is a random vector with arbitrary distribution with no positive probability mass in any proper linear subspace. 
For any $i\in X_1$, node $i$ is only connected to node $i+k$. Since $\theta_i = B$ and $\theta_{k+i}=z_i$ for a random variable $z_i$, $\theta_i -\theta_{k+i}\neq B$ almost surely. So in order for the flow equations to hold, either both $i, i+k\in H_0$ or $i+k\in H_0$. The same argument holds for any node $j\in Y_1$ and its only neighbor $j-3k$. So in order for the problem to have a solution, $H_0$ should contain both $X_2$ and $Y_2$. On the other hand, since $|V_{H_0}|\leq|V|/2$, therefore $H_0=G[X_2\cup Y_2]=H$ is the only possible attacked area. Now since $H_0=H$, we can assume that the attacked area is given and the rest of the proof is exactly similar to the proof of Lemma~\ref{lem:hardness_H}.
\end{proof}
\begin{corollary}\label{cor:hardness_cyber}
Given $\A, \tet,$ and $\tet^{\star}$, it is strongly NP-hard to find a subgraph $H$, a set of line failures $F$ in $H$, and a vector $\tet'_H$ such that $\A\tet = \A'\left[\begin{smallmatrix}\tet'_H\\\tet^{\star}_{\bar{H}}\end{smallmatrix}\right]$, even if such $H,F$ exist.
\end{corollary}
Corollary~\ref{cor:hardness_cyber} indicates that it is NP-hard to detect the line failures after an attack as described in Section~\ref{sec:Model} in general cases. However, in the next sections, we provide a polynomial-time algorithm to detect the attacked area $H$ and the set of line failures $F$, and show based on simulations that it performs well in reasonable scenarios. 
\section{Attacked Area Approximation}\label{sec:area}
In this section, we provide methods to approximate the attacked area after a cyber attack as described in Subsection~\ref{subsec:attack}. In subsections~\ref{subsec:data_distortion} and \ref{subsec:data_replay}, we first provide methods to contain the attacked area after the data distortion and replay attacks, respectively.
We then combine these methods in the ATtacked Area Containment (ATAC) Module for containing the attacked area after both types of data attacks. Finally, given an area containing the attacked area, in subsection~\ref{subsec:improve}, we provide methods to improve the approximation of the attacked area.
\subsection{Data Distortion}\label{subsec:data_distortion}
We first consider data distortion attacks. In particular, recall that we assume that $\tet_H^{\star}=\tet_H'+\vec{z}$ for a random vector $\vec{z}$ with an arbitrary distribution with no positive probability mass in any proper linear subspace. Since $\tet_H^{\star}$ is the vector of the modified phase angles and there are also some line failures in $H$, it can be seen that $\A\tet^{\star}\neq \vec{p}$.
\begin{lemma}\label{lem:Hbar_int}
For any $i\in \intt(\bar{H})$, $\A_i\tet^{\star}=p_i$.
\end{lemma}
\begin{proof}
Since $i\in \intt(\bar{H})$, therefore $a_{ij}=0$ for all $j\in V_H$. Hence, $\A_i\tet^{\star}=\A_{i|\bar{H}}\tet_{\bar{H}}^{\star}$. On the other hand, since the attack is inside $H$, we know $\A_{i|\bar{H}}=\A_{i|\bar{H}}'$, and also $\tet_{\bar{H}}^{\star}=\tet_{\bar{H}}'$. Hence, $\A_i\tet^{\star}=\A_{i|\bar{H}}\tet_{\bar{H}}^{\star}=\A_{i|\bar{H}}'\tet_{\bar{H}}'=p_i$.
\end{proof}
\begin{lemma}\label{lem:not_int}
For any $i\!\in\! V\backslash \intt(\bar{H})$, $\A_i\tet^{\star}\!\neq\! p_i$ almost surely.
\end{lemma}
\begin{proof}
For any $i\in V\backslash\intt(\bar{H})$, there exists a node $j\in V_H$ such that $a_{ij}\neq 0$. Now since the set of solutions $\vec{x}$ to $\A_i\vec{x}= p_i$ is a measure zero set in $\mathbb{R}^n$ and $\theta_j^{\star}$ is a random modification of $\theta_j'$, $\A_i\tet^{\star}\neq p_i$ almost surely.
\end{proof}
Lemmas~\ref{lem:Hbar_int} and \ref{lem:not_int} indicate that given $\A,\tet,$ and $\tet^{\star}$ one can find $\intt(\bar{H})$ by computing $V\backslash \supp(\A\tet^{\star}- \vec{p})$.
\begin{corollary}\label{cor:int_Hbar}
$\intt(\bar{H})\!=\!V\backslash \supp(\A\tet^{\star}\!-\! \vec{p})$, almost surely.
\end{corollary}

Define $S_0:=G[\supp(\A\tet^{\star}-\pvec{p})]$. We know from Corollary~\ref{cor:int_Hbar} that $\intt(\bar{H})=V_{\bar{S}_0}$ and from Lemma~\ref{lem:not_int} that $V_H\subset V_{S_0}$. Therefore, $S_0$ clearly contains $H$.
The following lemma demonstrates that $\intt(S_0)$ is a better approximation for $V_H$. We use this lemma in Subsection~\ref{subsec:improve} to improve the approximation of the attacked area.
\begin{lemma}\label{lem:estimate_H1}
$V_H\subseteq \intt(S_0)$, almost surely.
\end{lemma}
\begin{proof}
Assume not. Then there exists a node $i\in V_H$ such that $N(i)\cap V_{\bar{S}_0}\neq \emptyset$. Assume $j\in N(i)\cap V_{\bar{S}_0}\neq \emptyset$, then with a similar argument as in the proof of Lemma~\ref{lem:not_int}, one can show that $\A_j\tet^{\star}\neq p_j$ almost surely, which contradicts with $j\notin V_{{S}_0}$. Hence, $N(i)\cap V_{\bar{S}_0} = \emptyset$ and $V_H\subseteq \intt(S_0)$.
\end{proof}

\subsection{Data Replay}\label{subsec:data_replay}
In this subsection, we consider data replay attacks. Recall that we assume $\tet_H^{\star}=\tet_H''$ such that $\tet''$ satisfies $\A\tet'' =\pvec{p}''$. The power supply/demand vector $\pvec{p}''$ is arbitrarily selected such that $\pvec{p}_H'' =\pvec{p}_H$, and $\pvec{p}''_{\bar{H}}$ is selected generally enough.

The data replay attacks are harder to detect since the data seems to be correct locally. Again, one can easily see that $\A\tet^{\star}\neq \vec{p}$, but here unlike the data distortion case, not all the nodes in $H$ can be detected by checking $\A_i\tet^{\star}\neq p_i$. The following lemma shows why attacked area containment is more difficult after a data replay attack.
\begin{lemma}\label{lem:int_replay}
For any $i\in \intt(H)\cup \intt(\bar{H})$, $\A_i\tet^{\star}=p_i$.
\end{lemma}
\begin{proof}
Similar to the proof of Lemma~\ref{lem:Hbar_int}, it is easy to show that for any $i\in \intt(\bar{H})$, $\A_i\tet^{\star}=p_i$. The only new part is to show the same for nodes in $\intt(H)$. So assume $i\in\intt(H)$, following the definition of the interior, it can be verified that $\A_i\tet^{\star} = \A_{i|H}\tet^{\star}_H$. On the other hand, since $\tet^{\star}_H=\tet_H''$ and $\pvec{p}_H''=\pvec{p}_H$, we can verify that $\A_{i|H}\tet^{\star}_H =\A_{i|H}\tet_H''=p_i''=p_i$. Hence, for all $i\in \intt(H)$ also $\A_i\tet^{\star}=p_i$.
\end{proof}
\begin{lemma}\label{lem:boundary_replay}
 For any $i\in \partial(H)\cup\partial(\bar{H})$, $\A_i\tet^{\star}\neq p_i$, almost surely.
\end{lemma}
\begin{proof}
The proof of this lemma is similar to the proof of Lemma~\ref{lem:not_int}. For any $i\in \partial(\bar{H})$, there exists a node $j\in V_H$ such that $a_{ij}\neq 0$. Now since the set of solutions $\vec{x}$ to $\A_i\vec{x}= p_i$ is a measure zero set in $\mathbb{R}^n$ and $\theta_j^{\star}=\theta_j''$ for a generally enough selected vector $p_{\bar{H}}''$, $\A_i\tet^{\star}\neq p_i$ almost surely. A similar argument holds for $i\in \partial(H)$.
\end{proof}
\begin{corollary}\label{cor:S_replay}
$\supp(\A\tet^{\star}-\pvec{p})=\partial(H)\cup\partial(\bar{H})$, almost surely.
\end{corollary}
From comparing Corollaries~\ref{cor:int_Hbar} and \ref{cor:S_replay}, one can see that in the replay attack case, $S_0=G[\supp(\A\tet^{\star}-\pvec{p})]$ does not contain the attacked area $H$ anymore. In the following lemma, we show how one can still contain the attacked area in this case.
\begin{lemma}\label{lem:cc_H_Hbar}
If $C_1,C_2,\dots,C_k$ are the connected components of $G\backslash S_0$, then these connected components can be divided into two disjoint sets $\{i_1,i_2,\dots,i_s\}$ and $\{j_1,j_2,\dots,j_t\}$ such that $G[\intt(H)] = C_{i_1}\cup C_{i_2}\cup\dots\cup C_{i_s}$ and $G[\intt(\bar{H})] = C_{j_1}\cup C_{j_2}\cup\dots\cup C_{j_t}$.
\end{lemma}
\begin{proof}
It is a direct result of Corollary~\ref{cor:S_replay}.
\end{proof}
\begin{lemma}\label{lem:cc_common_node}
For two connected components $C_i$ and $C_j$ of $G\backslash S_0$, if $N(C_i)\cap N(C_j)\neq \emptyset$, then either $C_i\cup C_j\subseteq \intt(H)$ or $C_i\cup C_j\subseteq \intt(\bar{H})$.
\end{lemma}
\begin{proof}
From Lemma~\ref{lem:cc_H_Hbar}, for any $i$, either $C_i\subseteq G[\intt(\bar{H})]$ or $C_i\subseteq G[\intt(H)]$. If $C_i\subseteq G[\intt(\bar{H})]$ then $N(C_i)\subseteq \partial(\bar{H})$, and if $C_i\subseteq G[\intt(H)]$ then $N(C_i)\subseteq \partial(H)$. Hence, since $\partial(\bar{H})\cap\partial(H)=\emptyset$, if $N(C_i)\cap N(C_j)\neq \emptyset$, then either $C_i\cup C_j\subseteq \intt(H)$ or $C_i\cup C_j\subseteq \intt(\bar{H})$.
\end{proof}
Following Lemma~\ref{lem:cc_common_node}, one can see that connected components $C_1,C_2,\dots,C_k$ can be combined into disjoint subgraphs $G_1,\\G_2,\dots, G_t$ such that for any two of these subgraphs such as $G_i$ and $G_j$, $N(G_i)\cap N(G_j)=\emptyset$. Moreover for any $i$, either $G_i\subseteq G[\intt(H)]$ or $G_i\subseteq G[\intt(\bar{H})]$. In the following lemma, we use this fact to contain the attacked area.
\begin{lemma}\label{lem:approx_H_replay}
There exists $1\leq i\leq t$, such that $H\subset G\backslash G_i$. Moreover, $H\subseteq  G[\intt(G\backslash G_i)]$.
\end{lemma}
\begin{proof}
The first part of the proof is the direct result of Lemmas~\ref{lem:cc_H_Hbar} and \ref{lem:cc_common_node}. To prove the second part, notice that for any $i$, $S_0\subset G\backslash G_i$. Therefore, for any $i$, $\partial(\bar{H})\subset G\backslash G_i$. Hence, if $H\subset G\backslash G_i$, since $\partial(\bar{H})\subset G\backslash G_i$, one can verify that $H\subseteq  G[\intt(G\backslash G_i)]$.
\end{proof}

Lemma~\ref{lem:approx_H_replay} demonstrates that at least one of $G\backslash G_i$ contains the attacked area. Hence, one can use this fact to contain the attacked area after a data replay attack.

\begin{figure}[t]
\centering
\includegraphics[scale=0.4]{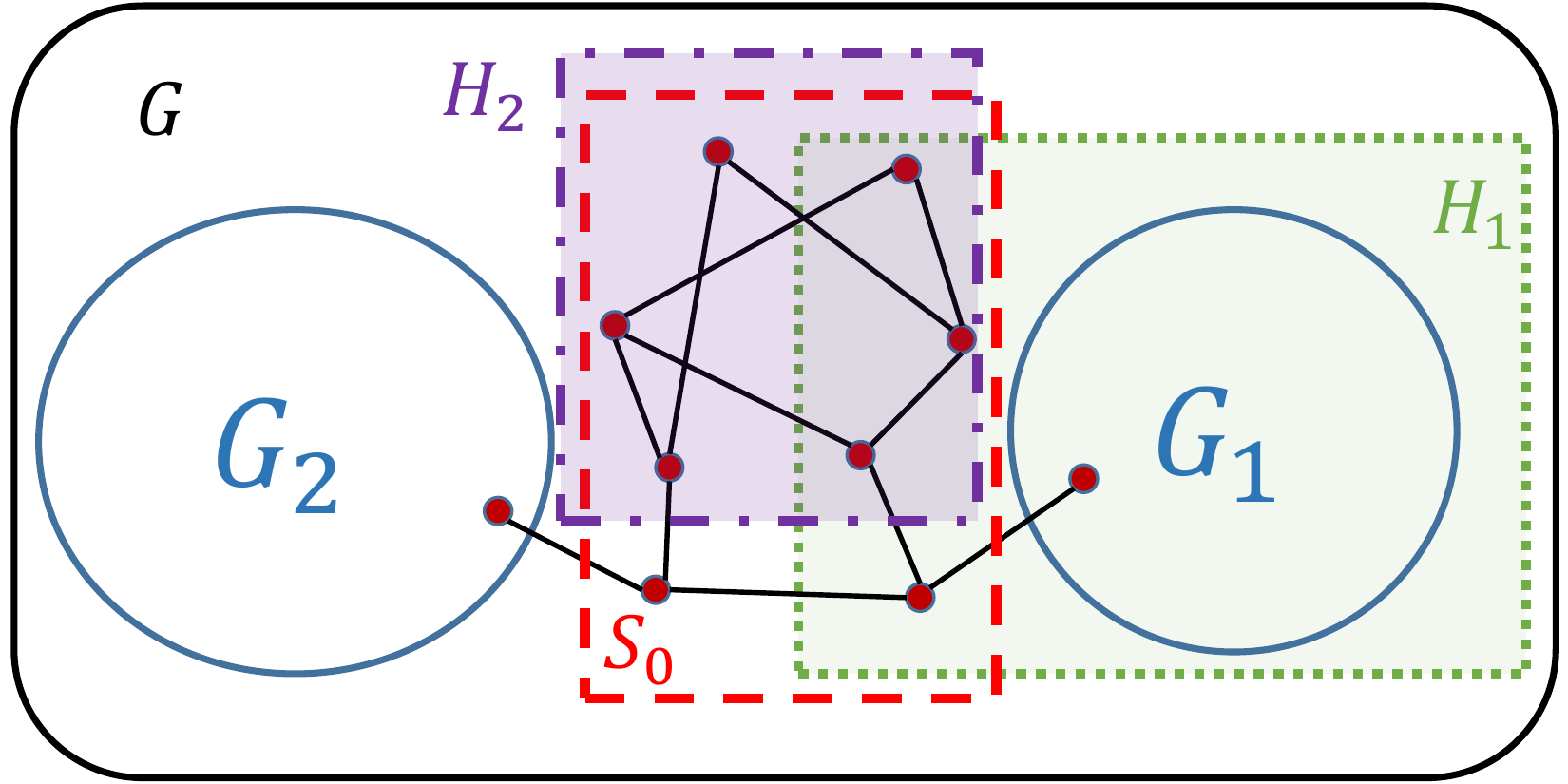}
\vspace*{-0.2cm}
\caption{An ambiguous scenario. Both a data replay attack on the attacked area $H_1$ or a data distortion attack on the attacked area $H_2$ result in the same  $S_0=G[\supp(\A\tet^{\star}-\vec{p})]$.}
\label{fig:Amb_scen}
\end{figure}

\subsection{The ATAC Module} \label{subsec:ATAC}
Using the results in the previous subsections, here we introduce the ATtacked Area Containment (ATAC) Module for containing the attacked area after both types of data attacks. The main challenge is to distinguish between the two data attacks. As shown in Fig.~\ref{fig:Amb_scen}, there are scenarios for which the data attack type cannot be recognized by simply looking at $S_0$. Hence, the ATAC Module does not return a single subgraph containing the attacked area but a series of possible subgraphs. In Sections~\ref{sec:line} and \ref{sec:algorithm}, we show that by defining the \emph{confidence of the solution}, an algorithm can go over all of these subgraphs until it detects the attacked area and the set of line failures  with high confidence.

The steps of the ATAC Module are summarized in Module~\ref{module:ATAC}. As can be seen, $S_0$ is the first possible subgraph returned by the ATAC module, which is for the case when there is a data distortion attack. Then based on Lemma~\ref{lem:approx_H_replay}, $S_1:=G\backslash G_1, S_2:=G\backslash G_2,\dots, S_t:=G\backslash G_t$ are other possible areas containing the attacked area, if there is a replay attack. Notice that since $t<|V|$, therefore the ATAC module is a polynomial time algorithm.

\begin{module}[t]
\caption{ATtacked Area Containment (ATAC)}
\label{module:ATAC}
\small
\begin{trivlist}
\item\textbf{Input:} $G$, $\A$, $\tet$, and $\tet^{\star}$
\end{trivlist}
\vspace*{-3mm}
\begin{algorithmic}[1]
\STATE Compute $\vec{p} = \A\tet$
\STATE Compute $S_0 = G[\supp(\A\tet^{\star}- \vec{p})]$
\STATE Find the connected components $C_1,C_2,\dots,C_k$ of $G\backslash S_0$
\STATE Using Lemma~\ref{lem:cc_common_node}, combine the connected components with common neighbors to obtain $G_1,\dots,G_t$ (sorted based on their size from largest to smallest)
\STATE Return $S_0, S_1:=G\backslash G_1, S_2:=G\backslash G_2,\dots,S_t:=G\backslash G_t$
\end{algorithmic}
\end{module}

\subsection{Improving Attacked Area Approximation}\label{subsec:improve}
Assume that from the subgraphs returned by the ATAC Module, $S^*$ is one of them that contains the attacked area $H$. Following Lemma~\ref{lem:estimate_H1} and Lemma~\ref{lem:approx_H_replay}, $S_a:=G[\intt(S^*)]$ is a better approximation for the attacked area $H$.  In order to find a more accurate approximation for $H$, we provide the following lemma which is similar to \cite[Lemma 1]{SYZ2015}.

\begin{lemma}\label{lem:StoH}
For a subgraph $S$, if $V_H\subseteq V_S$, then:
\begin{equation}\label{eq:StoH}
\A_{\bar{S}|G} (\tet-\tet') = 0.
\end{equation}
\end{lemma}
\begin{proof}
Since all the line failures are inside $H$, and also $V_H\subseteq V_S$, therefore it can be seen that $\A_{\bar{S}|G} = \A_{\bar{S}|G}'$. On the other hand, $\A_{\bar{S}|G}\tet = \vec{p}_{\bar{S}}$ and $\A_{\bar{S}|G}'\tet' = \vec{p}_{\bar{S}}$. Hence, $\A_{\bar{S}|G}\tet - \A_{\bar{S}|G}'\tet'=0$ and therefore $\A_{\bar{S}|G} (\tet-\tet') = 0$.
\end{proof}

Lemma~\ref{lem:StoH} can be effectively used to estimate the phase angle of the nodes in $S$ and to detect the attacked area $H$ using these estimated values. The idea is to break (\ref{eq:StoH}) into parts that are known and unknown as follows:
\begin{align*}
\A_{\bar{S}|\bar{S}} (\tet_{\bar{S}}-\tet'_{\bar{S}}) + \A_{\bar{S}|S} (\tet_{S}-\tet'_{S}).
\end{align*}
Notice that since $V_H\subseteq V_S$, therefore $\tet'_{\bar{S}}=\tet^{\star}_{\bar{S}}$. Hence, the only unknown variable in the equation above is $\tet'_{S}$. Assume $\vec{y}\in \mathbb{R}^{|V_S|}$ is a solution to the following equation:
\begin{equation}\label{eq:H_detect}
\A_{\bar{S}|S} \vec{y} = \A_{\bar{S}|\bar{S}} (\tet_{\bar{S}}-\tet^{\star}_{\bar{S}}) + \A_{\bar{S}|S} \tet_{S}.
\end{equation}
In the following lemma, we demonstrate that $\supp(\vec{y}-\tet^{\star}_{S})$ can be used to estimate $H$.
\begin{lemma}\label{lem:estimate_H2}
If $\vec{y}$ is a solution to (\ref{eq:H_detect}), $V_H\subseteq \supp(\vec{y}-\tet^{\star}_{S})$, almost surely.
\end{lemma}
\begin{proof}
Since $\tet^{\star}_H$ is selected generally enough (for both the data distortion and replay attacks) for any $i\in H$, the only way $y_i = \theta_i^{\star}$ satisfying (\ref{eq:H_detect}) is that $\A_{\bar{S}|i}=0$. In that case any $y_i\in \mathbb{R}$ satisfies (\ref{eq:H_detect}). So the set of solutions $\vec{y}$ such that $y_i = \theta_i^{\star}$ is a measure zero set and $y_i \neq \theta_i^{\star}$  almost surely. Hence, $V_H\subseteq \supp(\vec{y}-\tet^{\star}_{S})$, almost surely.
\end{proof}
Following Lemma~\ref{lem:estimate_H2}, if $\vec{y}$ is a solution to (\ref{eq:H_detect}) for $S=S_a$, then $S_b:=G[\supp(\vec{y}-\tet^{\star}_{S_a})]$ is a better approximation for the attacked area $H$. Fig.~\ref{fig:attack_containment} shows the difference between $S^*$, $S_a$, and $S_b$ in approximating the attacked area for the case of a distortion attack.

\begin{figure}[t]
\centering
\includegraphics[scale=0.3]{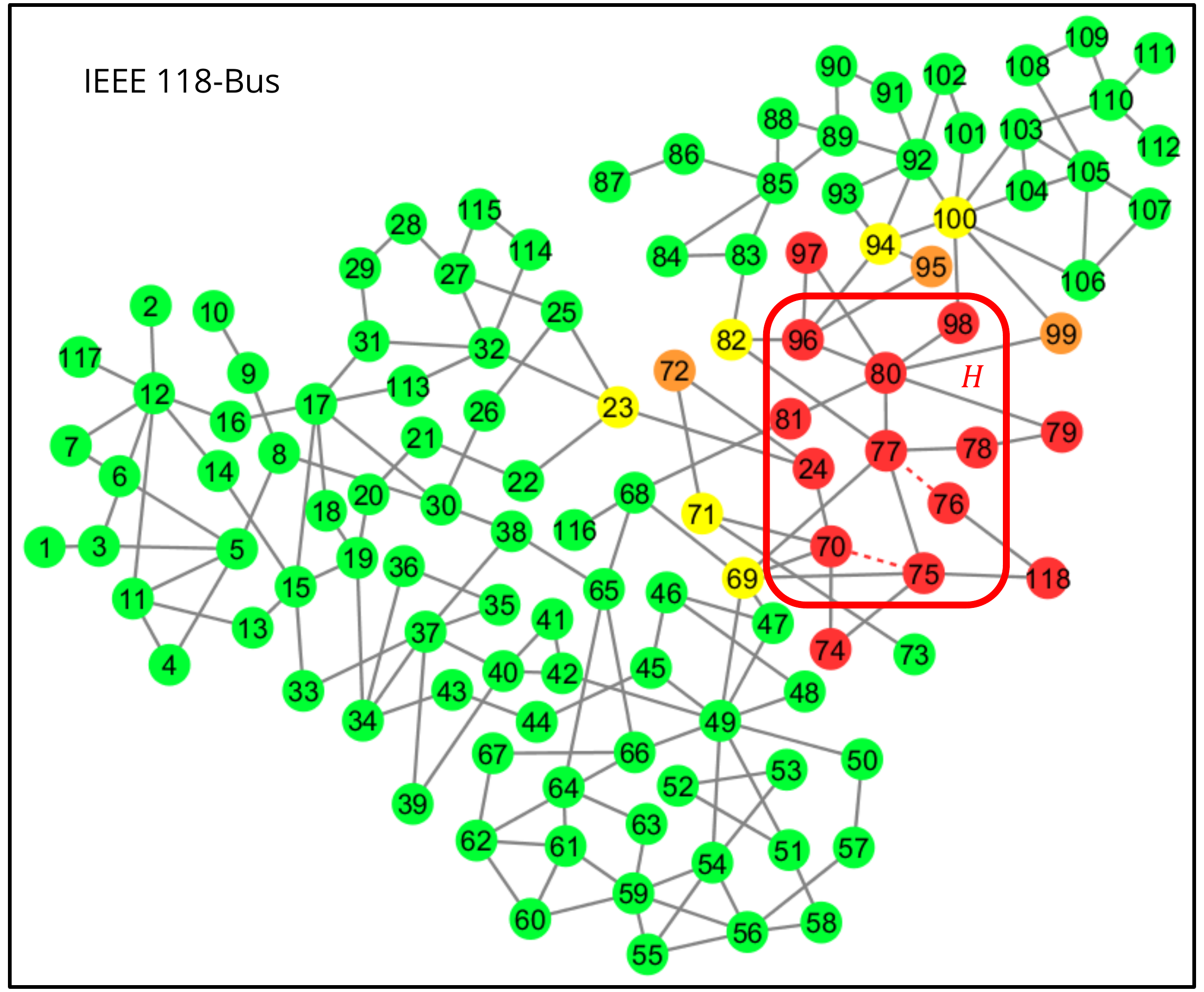}
\vspace*{-0.2cm}
\caption{$H$ is an induced subgraph of $G$ that represents the  attacked area with data distortion attack. The red, orange, and yellow nodes represent the nodes in $S_b$, $S_a\backslash S_b$, and $S^*\backslash S_a$, respectively.}
\label{fig:attack_containment}
\end{figure}

Finally, the following lemma demonstrates when $S_b$ is exactly equal to $H$.
\begin{lemma}\label{lem:exact}
For a subgraph $S$ such that $V_H\subseteq V_S$, if $V_S\backslash V_H\subseteq \partial(S)$ and there is a matching between the nodes in $\bar{S}$ and $\partial(S)$ that covers all the nodes in $\partial(S)$, then $G[\supp(\vec{y}-\tet^{\star}_{S})]=H$, in which $\vec{y}$ is the solution to (\ref{eq:H_detect}).
\end{lemma}
\begin{proof}
If there is a matching between the nodes inside and outside of $H$ that covers all the nodes in $\partial(S)$, one can prove that $\A_{\bar{S}|\partial(S)}$ has linearly independent columns, almost surely (see~\cite[Corollary 2]{SYZ2015}). Moreover, it is easy to see that $\A_{\bar{S}|\intt(S)}=0$. Hence, if $\vec{y}$ is a solution to (\ref{eq:H_detect}), $y_{\partial(H)}=\tet'_{\partial(H)}$. Now since $V_S\backslash V_H\subseteq \partial(S)$, for any $i$ in $V_S\backslash V_H$, $y_i = \theta'_i=\theta^{\star}_i$. On the other hand, since $\tet^{\star}_H$ are selected generally enough, one can verify that for any $i\in H$, $y_i\neq\tet^{\star}_i$, almost surely. Therefore, $G[\supp(\vec{y}-\tet^{\star}_{S})]=H$, almost surely.
\end{proof}

\section{Line Failures Detection}\label{sec:line}
In the previous section, we provided methods to find a good approximation $S$ for the attacked area $H$. In this section, we provide a method to detect line failures inside $S$. For this reason,  we use and build on the idea introduced in~\cite{SYZ2015}. It was proved in~\cite{SYZ2015} that if the attacked area $H$ is known, then there always exists feasible vectors $\vec{x}\in\mathbb{R}^{|E_H|}$ and $\vec{y}\in\mathbb{R}^{|V_H|}$ satisfying the conditions of the following optimization problem such that $\supp(\vec{x})=F$ and $\vec{y}=\tet_H'$:
\begin{eqnarray}\label{eq:simul_detect}
&&\min_{\vec{x},\vec{y}} \|\vec{x}\|_1~\text{s.t.}\nonumber\\
&&\A_{H|H}(\tet_H-\vec{y})+\A_{H|\bar{H}}(\tet_{\bar{H}}-\tet'_{\bar{H}})=\D_H\vec{x}\\
&&\A_{\bar{H}|H}(\tet_H-\vec{y})+\A_{\bar{H}|\bar{H}}(\tet_{\bar{H}}-\tet'_{\bar{H}})=0.\nonumber
\end{eqnarray}

Notice that  the  optimization problem~(\ref{eq:simul_detect}) can be solved efficiently using Linear Program (LP). It is proved in~\cite{SYZ2015} that under some conditions on $H$ and the set of line failures $F$, the solution to~(\ref{eq:simul_detect}) is unique, therefore the relaxation is exact and the set of line failures can be detected by solving~(\ref{eq:simul_detect}). In particular, it is proved in~\cite{SYZ2015} that if $H$ is acyclic and there is a matching between the nodes in $H$ and $\bar{H}$ that covers $H$, the solution to~(\ref{eq:simul_detect}) is unique for any set of line failures.

Since the conditions on $H$ and $F$ as described in~\cite{SYZ2015} may not always hold for the exactness of the line failures detection using (\ref{eq:simul_detect}), it cannot be used in general cases to detect line failures. To address this issue, here, we introduce a randomized version of (\ref{eq:simul_detect}).

Assume that $\W \in \mathbb{R}^{|E_S|\times |E_S|}$ is a diagonal matrix. We show that the solution to the following optimization problem can detect line failures in $S$ accurately for a ``good" matrix $\W$:
\begin{eqnarray}\label{eq:weighted_simul_detect}
&&\min_{\vec{x},\vec{y}} \|\W\vec{x}\|_1~\text{s.t.}\nonumber\\
&&\A_{S|S}(\tet_S-\vec{y})+\A_{S|\bar{S}}(\tet_{\bar{S}}-\tet'_{\bar{S}})=\D_S\vec{x}\\
&&\A_{\bar{S}|S}(\tet_S-\vec{y})+\A_{\bar{S}|\bar{S}}(\tet_{\bar{S}}-\tet'_{\bar{S}})=0.\nonumber
\end{eqnarray}
The idea behind optimizing the weighted norm-1 of vector $\vec{x}$ is to be able to detect the line failures when the solution to (\ref{eq:simul_detect}) does not detect the correct set of line failures but a small disturbance results in the correct detection.

\begin{module}[t]
\caption{LIne Failures Detection (LIFD)}
\label{module:LIFD}
\small
\begin{trivlist}
\item\textbf{Input:} $G$, $\A$, $\tet$, $S$, $T$, and $\tet^{\star}$
\end{trivlist}
\vspace*{-3mm}
\begin{algorithmic}[1]
\STATE Compute $\vec{p} = \A\tet$
\STATE Compute a solution $\vec{x},\vec{y}$ to (\ref{eq:weighted_simul_detect}) for $\W =\I$
\STATE Set $F^{\dagger} =\supp(\vec{x})$ and $\pvec{\theta}_S^{\dagger}=\vec{y}$
\WHILE {$c(F^{\dagger},\pvec{\theta}_S^{\dagger})<99.99\%$ \& counter<T}
    \STATE counter++
    \STATE Draw random numbers $w_1,w_2,\dots,w_{|V_S|}$ from an exponential distribution with rate $\lambda = 1$
    \STATE Compute a solution $\vec{x},\vec{y}$ to (\ref{eq:weighted_simul_detect}) for $\W=\diag(w_1,w_2,\dots,w_{|V_S|})$
    \STATE Set $F^{\dagger} =\supp(\vec{x})$ and $\pvec{\theta}_S^{\dagger}=\vec{y}$
\ENDWHILE
\IF {$c(F^{\dagger},\pvec{\theta}_S^{\dagger})>99.99\%$}
    \STATE \textbf{return} $F^{\dagger},\pvec{\theta}_S^{\dagger}$
\ELSE
    \STATE \textbf{return} $F^{\dagger},\pvec{\theta}_S^{\dagger}$ with maximum $c(F^{\dagger},\pvec{\theta}_S^{\dagger})$ in all iterations
\ENDIF
\end{algorithmic}
\end{module}

Before we demonstrate the effectiveness of the optimization (\ref{eq:weighted_simul_detect}) in detecting line failures, we provide a metric for measuring the \emph{confidence of a solution}. In a subgraph $S$, assume $F^{\dagger}=\supp(\vec{x})$ and $\pvec{\theta}_S^{\dagger}=\vec{y}$ are the set of detected line failures and the recovered phase angles using the solution to (\ref{eq:weighted_simul_detect}). Also assume that $\A^{\dagger}$ is the admittance matrix after removing the lines in $F^{\dagger}$ and define $\pvec{p}^{\dagger}:=\A_{G|\bar{S}}\tet'_{\bar{S}}+\A_{G|S}^{\dagger}\pvec{\theta}^{\dagger}_S$. Notice that $\vec{x}$ and $\vec{y}$ satisfying (\ref{eq:weighted_simul_detect}) does not necessarily imply $\pvec{p}^{\dagger}=\pvec{p}$. Hence, one can use this difference to compute the confidence of a solution as follows.
\begin{definition}
 The \emph{confidence of the solution} is denoted by $c(F^{\dagger},\pvec{\theta}_S^{\dagger})$ and defined as:
\begin{equation}\label{eq:confidence}
c(F^{\dagger},\pvec{\theta}_S^{\dagger}) := (1-\|\pvec{p}^{\dagger} - \vec{p}\|_2/\|\vec{p}\|_2)^+\times100,
\end{equation}
in which $(z)^+ := \max(0,z)$.
\end{definition}

The confidence of the solution along with a random selection of the weight matrix $\W$ in (\ref{eq:weighted_simul_detect}) can be used to detect line failures that cannot be detected using (\ref{eq:simul_detect}). The idea is to repeatedly solve (\ref{eq:weighted_simul_detect}) using a random weight matrix  until the confidence of the solution for $F^{\dagger}=\supp(\vec{x})$ and $\pvec{\theta}_S^{\dagger}=\vec{y}$ is 100\% or reach a maximum number of iterations ($T$). Here, we consider the case when the diagonal entries of matrix $\W$ are randomly selected from an exponential distribution. This approach is summarized in Module~\ref{module:LIFD} as the LIne Failures Detection (LIFD) Module.

 Through the rest of this section, we demonstrate why the LIFD Module is effective and when the number of iterations ($T$) is enough to be polynomial in terms of the input size to make sure that it finds the line failures accurately.  

\begin{lemma}\label{lem:expdistbound}
Assume $w_1,w_2,\dots,w_m$ are i.i.d.\ exponential random variables, then for $1\leq k\leq m-1$:
\begin{equation*}
Pr(\sum_{i=1}^k w_i<\sum_{i=k+1}^m w_i)= \frac{\sum_{j=k}^{m-1}\binom{m-1}{j}}{2^{m-1}}.
\end{equation*}
\end{lemma}
\begin{proof}
See Section~\ref{sec:proofs}.
\end{proof}
\begin{corollary}\label{cor:expdistbound}
Assume $w_1,w_2,\dots,w_m$ are i.i.d.\ exponential random variables, then for $k \leq m/2+\Theta(\sqrt{m})$:
\begin{equation*}
Pr(\sum_{i=1}^k w_i<\sum_{i=k+1}^m w_i)= \Omega(\frac{1}{\sqrt{m}}).
\end{equation*}
\end{corollary}
\begin{proof}
See Section~\ref{sec:proofs}.
\end{proof}

\begin{lemma}\label{lem:cycle}
If $S=H$, $H$ is a cycle with $m$ nodes and edges, and there is a matching between the nodes inside and outside of $H$ that covers all the inside nodes, then any set of line failures of size $k$ can be found by the LIFD Module for expectedly $T = 2^{m-1}/(\sum_{j=k}^{m-1}\binom{m-1}{j})$. Moreover, if $k \leq m/2+\Theta(\sqrt{m})$, then LIFD Module can detect line failures for $T = O(\sqrt{m})$.
\end{lemma}
\begin{proof}
First, one can see that if $S=H$, and there is a matching between the nodes inside and outside of $H$ that covers all the inside nodes, then $\A_{\bar{S}|S} = \A_{\bar{H}|H}$ has uniquely independent columns, almost surely~\cite[Corollary 2]{SYZ2015}. Hence, the solution $\vec{y}$ to (\ref{eq:weighted_simul_detect}) is unique and $\vec{y} = \tet'_H$. Therefore, we can assume that $\tet'$ is given. Without loss of generality assume that $F = \{e_1,\dots,e_k\}$. We prove that the solution $\vec{x}$ to (\ref{eq:weighted_simul_detect}) is unique and $\supp(\vec{x})=F$, if $\sum_{i=1}^k w_i<\sum_{i=k+1}^m w_i$, in which $\W=\diag(w_1,\dots,w_m)$.

\noindent Without loss of generality, assume that $\D_H$ is the incidence matrix of $H$ when lines of $H$ are oriented clockwise. Since $H$ is connected, it is known that $\rank(\D_H)=m-1$~\cite[Theorem 2.2]{bapat2010graphs}. Therefore, $\dim(\nul(\D_H))=1$. Suppose $\vec{z}\in \mathbb{R}^{|E_H|}$ is the all one vector. It can be verified that $\D_H\vec{z}=0$. Since $\dim(\nul(\D_H))=1$, $\vec{z}$ forms a basis for the null space of $\D$. Now suppose $\vec{x}^{\dagger}$ is a solution to $\A_{H|G}(\tet-\tet')=\D_H\vec{x}$ such that $\supp(\vec{x}^{\dagger})=F$ (from~\cite[Lemma 2]{SYZ2015}, we know that such a solution exists). Since $\vec{z}$ forms a basis for $\nul(D)$, all other solutions of $\A_{H|G}(\tet-\tet')=\D_H\vec{x}$ can be written in the form of $\vec{x}^{\dagger}+c\vec{z}$. We want to prove that if $\sum_{i=1}^k w_i<\sum_{i=k+1}^m w_i$, then for any $c\in \mathbb{R}\backslash\{0\}$, $\|\W\vec{x}^{\dagger}\|_1< \|\W(\vec{x}^{\dagger}+c\vec{z})\|_1$. Since $\supp(\vec{x}^\dagger)=F$, $x_1^{\dagger},x_2^{\dagger},\dots,x_k^{\dagger}$ are the only nonzero elements of $\vec{x}^\dagger$. Moreover $W_d : = \sum_{i=k+1}^m w_i-\sum_{i=1}^k w_i>0$. Hence,
\begin{align*}
\|\W(\vec{x}^\dagger+c\vec{z})\|_1 &= \sum_{i=1}^k w_i|x_i^{\dagger}-c|+|c| \sum_{i=k+1}^m w_i\\
&=\sum_{i=1}^k w_i(|x_i^{\dagger}-c|+|c|)+|c| W_d\\
&\geq\sum_{i=1}^k w_i |x_i^{\dagger}|+|c| W_d>\sum_{i=1}^k w_i |x_i^{\dagger}|=\|\W\vec{x}^{\dagger}\|_1.
\end{align*}
Therefore, the solution $\vec{x}$ to (\ref{eq:weighted_simul_detect}) is unique and $\supp(\vec{x})=F$, if $\sum_{i=1}^k w_i<\sum_{i=k+1}^m w_i$. One the other hand, from Lemma~\ref{lem:expdistbound}, $Pr(\sum_{i=1}^k w_i<\sum_{i=k+1}^m w_i)= \frac{\sum_{j=k}^{m-1}\binom{m-1}{j}}{2^{m-1}}$. Hence, expectedly $\frac{2^{m-1}}{\sum_{j=k}^{m-1}\binom{m-1}{j}}$ number of iterations ($T$) should be enough to satisfy this inequality. Corollary~\ref{cor:expdistbound} also gives the expected number of iterations needed when $k \leq m/2+\Theta(\sqrt{m})$.
\end{proof}
Lemma~\ref{lem:cycle} clearly demonstrates the effectiveness of using a weight matrix $\W$ in (\ref{eq:weighted_simul_detect}). It was previously proved in~\cite{SYZ2015} that if $H$ is a cycle and there is a matching between the nodes inside and outside of $H$ that covers all the inside nodes, then for any set of line failures of size \emph{less than half of the lines in $H$}, $\supp(\vec{x})$ of the solution $\vec{x}$ to (\ref{eq:simul_detect}) exactly reveals the set of line failures. However, for the line failures with the size more than half of the lines in $H$, this approach comes short. In these cases, Lemma~\ref{lem:cycle} indicates that solving (\ref{eq:weighted_simul_detect}) for random matrices $\W$ for polynomial number of times can lead to the correct detection.

\begin{algorithm}[t]
\caption{REcurrent Attack Containment and deTection (REACT)}
\label{algorithm:REACT}
\small
\begin{trivlist}
\item\textbf{Input:} $G$, $\A$, $\tet$, $\tet^{\star}$, and $T$
\end{trivlist}
\vspace*{-3mm}
\begin{algorithmic}[1]
\STATE Compute $\vec{p} = \A\tet$
\STATE Obtain $S_0, S_1,\dots, S_t$ using the ATAC Module
\FOR {$i=1$ to $t$}
    \STATE Compute $S_a = G[\intt(S_i)]$
    \IF{(\ref{eq:H_detect}) is feasible for $S = S_a$ }
        \STATE Find a solution $\vec{y}$ to (\ref{eq:H_detect}) for $S = S_a$
    \ELSE
        \STATE \textbf{continue}
    \ENDIF
    \STATE Compute $S_b = G[\supp(\vec{y}-\tet_S^{\star})]$
    \STATE Set $S = S_b$ as an approximation for the attacked area $H$
    \STATE Compute a solution $\vec{x},\vec{y}$ to (\ref{eq:weighted_simul_detect}) for $\W =\I$
    \STATE Set $F^{\dagger} =\supp(\vec{x})$ and $\pvec{\theta}_S^{\dagger}=\vec{y}$
    \IF{$c(F^{\dagger},\pvec{\theta}_S^{\dagger})<99.99\%$}
        \STATE Obtain $F^{\dagger},\pvec{\theta}_S^{\dagger}$ from module LIFD for inputs $S$ and $T$
    \ENDIF
    \IF{$c(F^{\dagger},\pvec{\theta}_S^{\dagger})>99.99\%$}
        \STATE \textbf{return} $H = \supp(\pvec{\theta}_S^{\dagger}-\pvec{\theta}_S^{\star})$ as the detected attacked area and $F^{\dagger},\pvec{\theta}_H^{\dagger}$ as the detected line failures and recovered phase angle of the nodes inside $H$
    \ENDIF
\ENDFOR
\STATE \textbf{return} $S$ and $F^{\dagger},\pvec{\theta}_S^{\dagger}$ with maximum $c(F^{\dagger},\pvec{\theta}_S^{\dagger})$ in all iterations
\end{algorithmic}
\end{algorithm}

Although providing a similar analytical bound for $T$ to ensure detecting line failures in general cases is very difficult, in Section~\ref{sec:numres}, we numerically show that small values of $T$ is enough to detect line failures in more complex attacked areas as well.

\section{REACT Algorithm}\label{sec:algorithm}
In this section, we present the REcurrent Attack Containment and deTection (REACT) Algorithm based on the results presented in the previous sections. The steps of the REACT Algorithm are summarized in Algorithm~\ref{algorithm:REACT}.

The REACT Algorithm first obtains a set of possible subgraphs $S_0, S_1,\dots, S_t$ that may contain the attacked area $H$ using the ATAC Module. Then, for each subgraph $S_i$ using the results in Subsection~\ref{subsec:improve}, it improves the approximation of the attacked area. In particular, it first computes $S_a = G[\intt(S_i)]$ and then finds a solution  to (\ref{eq:H_detect}) for $S=S_a$. If (\ref{eq:H_detect}) is not feasible, then it means that $S_i$ does not contain the attacked area $H$, and therefore, the algorithm goes to the next iteration and tries the next possible subgraph. If (\ref{eq:H_detect}) has a feasible solution $\vec{y}$, it obtains a better approximation of the attacked area $H$ by computing $S_b= G[\supp(\vec{y}-\tet_S^{\star})]$ (Lemma~\ref{lem:estimate_H2}). 

Then, it solves the optimization (\ref{eq:weighted_simul_detect}) for $\W=\I$, in which $\I$ is the identity matrix. Notice that this is basically similar to solving (\ref{eq:simul_detect}). Then it checks the confidence of the solution $c(F^{\dagger},\pvec{\theta}_S^{\dagger})$. If it is less than $99.99\%$, it calls the LIFD Module to obtain another solution $F^{\dagger},\pvec{\theta}_S^{\dagger}$. Finally, it checks whether the confidence of the solution is $c(F^{\dagger},\pvec{\theta}_S^{\dagger})>99.99\%$. If so, it approximates the attacked area $H$ using this solution and returns $F^{\dagger},\pvec{\theta}_H^{\dagger}$.

If the REACT Algorithm cannot find a solution with confidence greater than $99.99\%$, it returns a solution with the highest confidence between all the solutions obtained in all the iterations.

Notice that the REACT Algorithm is a polynomial time algorithm. Therefore, it cannot return the correct solution to an NP-hard problem in all cases. However, in the next section we numerically demonstrate that it performs very well in reasonable settings.

\begin{figure*}[t]
\centering
\includegraphics[scale=0.45]{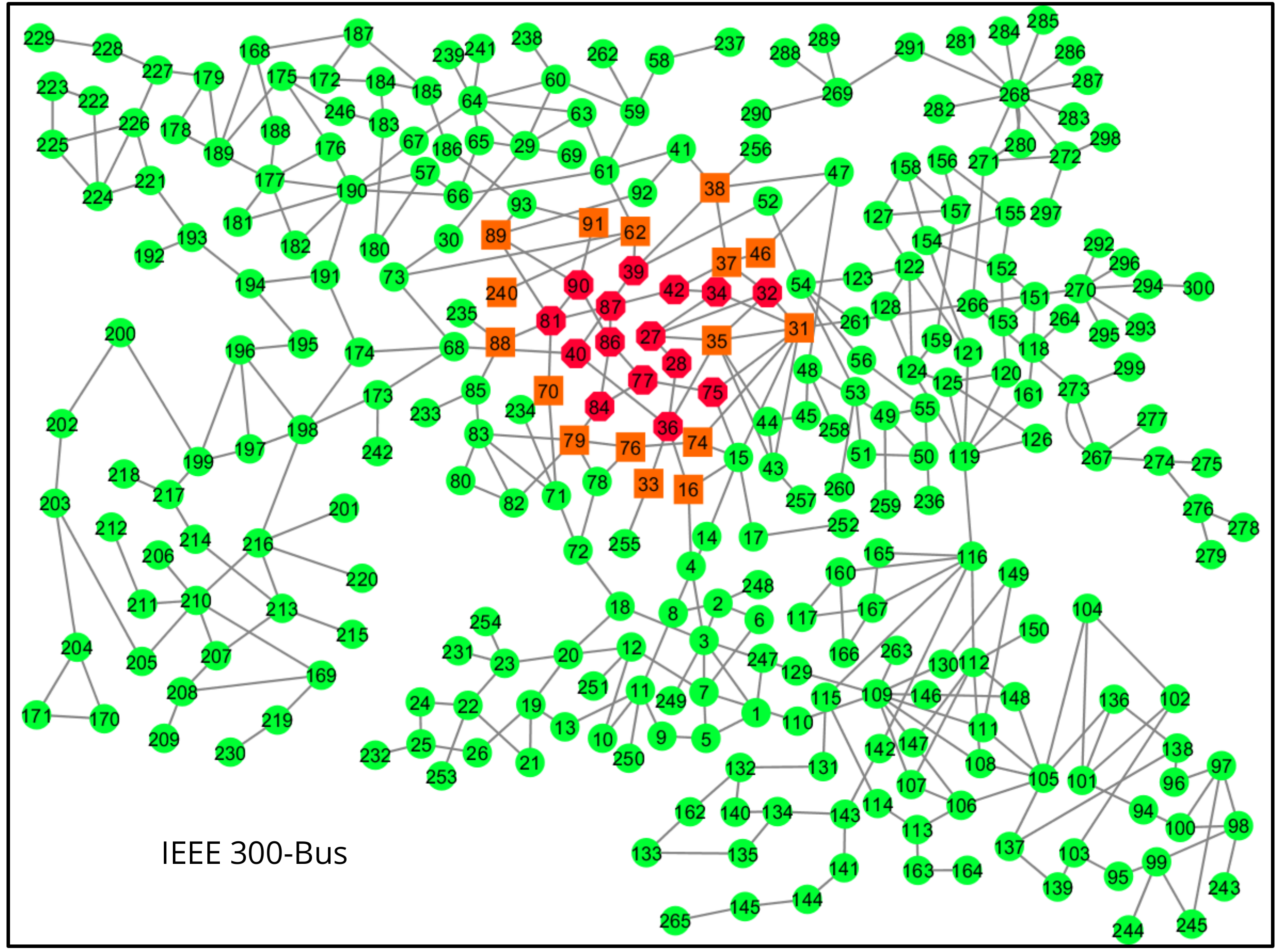}
\vspace*{-0.2cm}
\caption{The two attacked areas in the IEEE 300-bus systems that are used in simulations. The red octagon nodes are the nodes in $H_1$ and $H_2$, and the orange square nodes are the nodes that are only in $H_2$.}
\vspace*{-0.4cm}
\label{fig:attack_zones}
\end{figure*}



\section{Numerical Results}\label{sec:numres}
In this section, we evaluate the performance of the REACT Algorithm in detecting the attacked area and recovering the information after a cyber attack as described in Section~\ref{subsec:attack}. We consider two attacked areas $H_1$ and $H_2$ within the IEEE 300-bus system~\cite{IEEEtestcase} as depicted in Fig.~\ref{fig:attack_zones}. $H_1$ has 15 nodes and 16 edges, and $H_2$ which contains $H_1$, has 31 nodes and 41 edges. It can be verified that none of these two subgraphs are acyclic and there is no matching between the nodes inside and outside of these two subgraphs that covers their insides nodes. \emph{Hence, the methods provided in~\cite{SYZ2015} cannot recover the information inside these areas even when the attacked areas are known in advance.}

For the physical part of the attack, we consider all single line failures, and 100 samples of all double and triple line failures within $H_1$ and $H_2$. Figs.~\ref{fig:eval_H1} and \ref{fig:eval_H2} illustrate the REACT Algorithm's performance after these attacks. In the Algorithm, we set $T=20$ so that the while loop in the LIFD Module runs only for 20 iterations.

\begin{figure*}[t]
\centering
\begin{subfigure}{0.45\textwidth}
\centering
\includegraphics[scale=0.45]{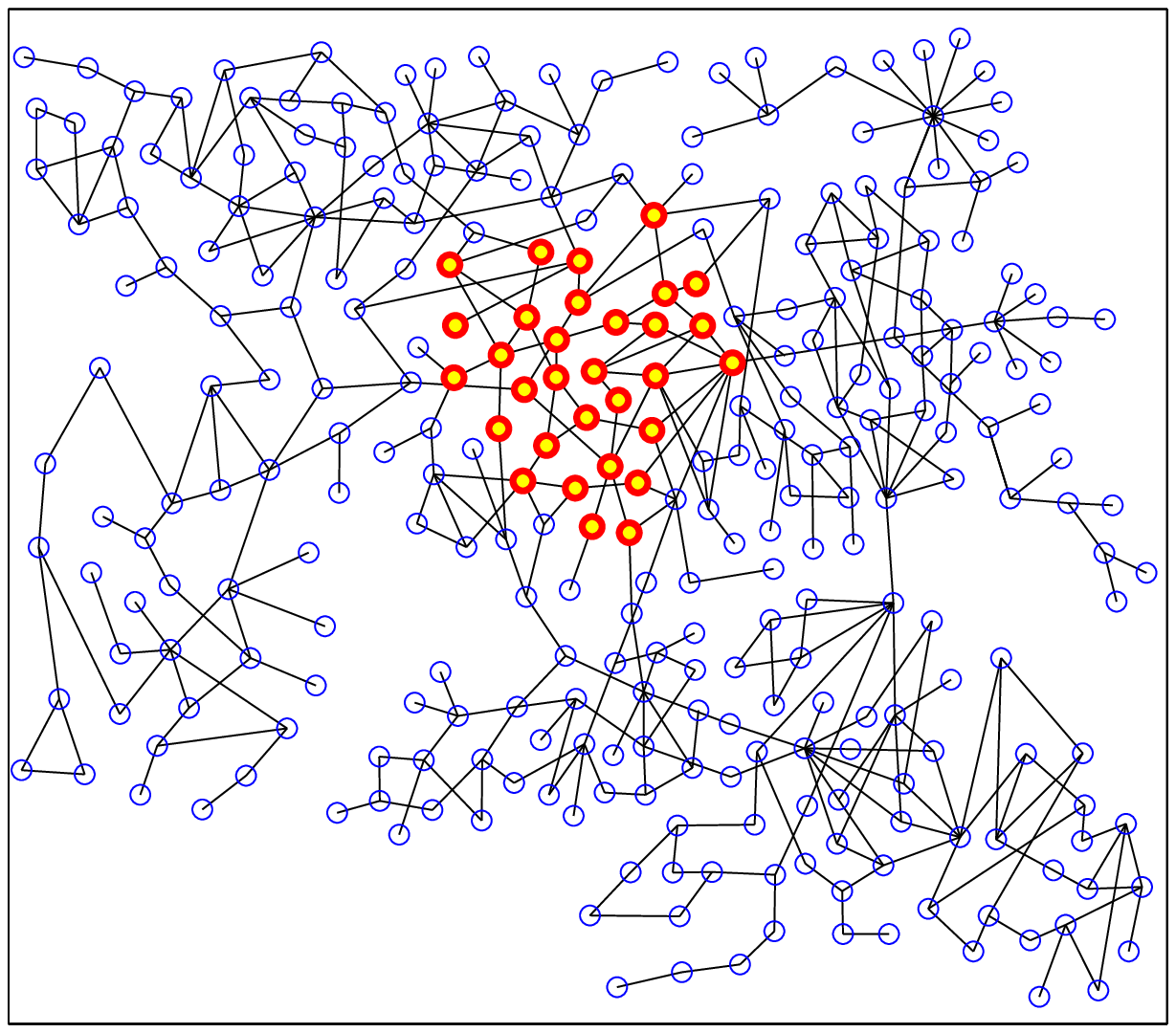}
\vspace*{-0.7cm}
\caption{Data Distortion Attack}
\label{fig:dist}
\end{subfigure}
\begin{subfigure}{0.45\textwidth}
\centering
\includegraphics[scale=0.45]{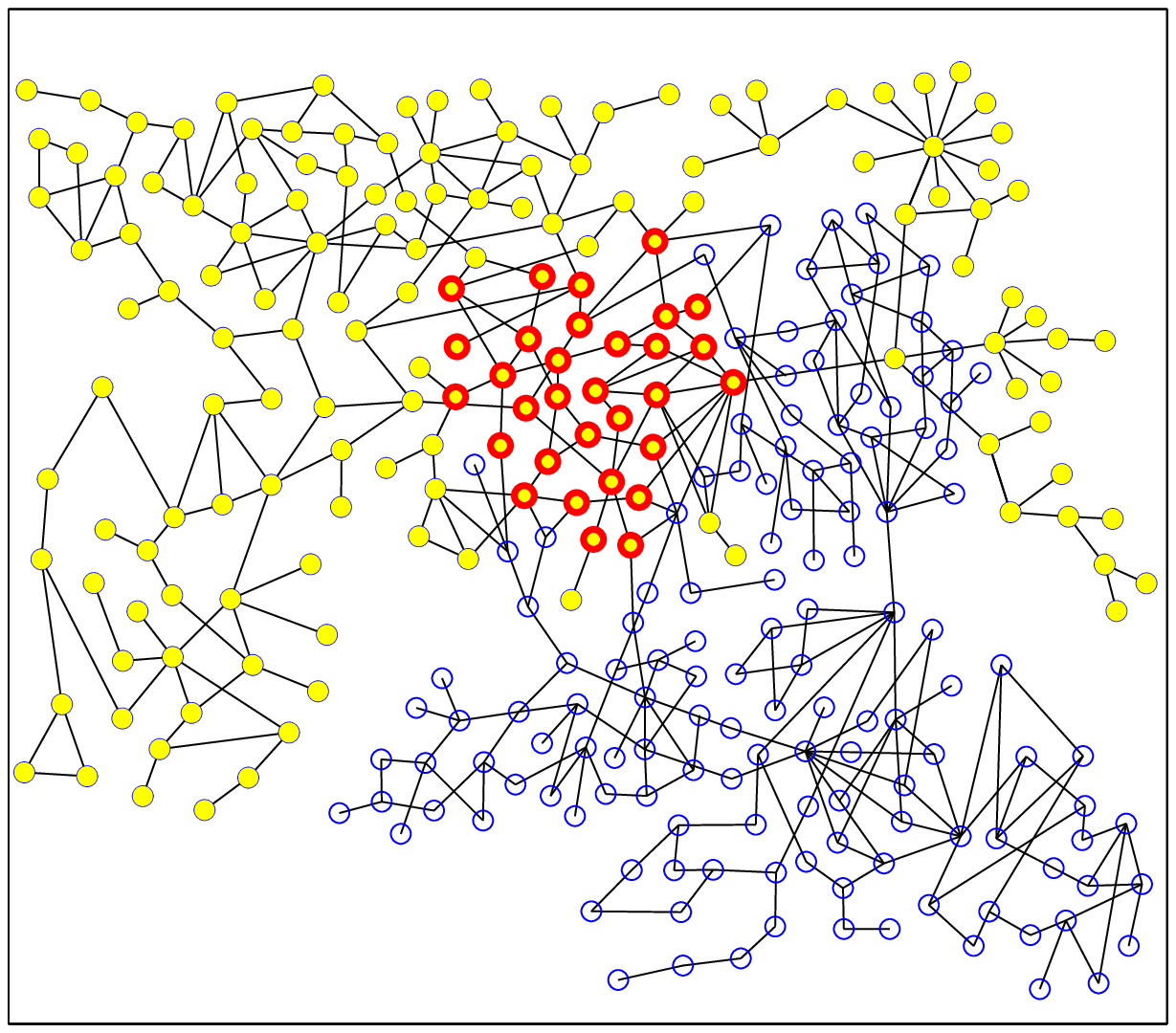}
\vspace*{-0.7cm}
\caption{Data Replay Attack}
\label{fig:rep}
\end{subfigure}
\caption{The difference in difficulty of detecting the attacked area after a data distortion attack and a data replay attack on the attacked area $H_2$ accompanied by a triple line failure within $H_2$. The yellow filled nodes represent the nodes in the detected attacked area by the REACT Algorithm, the nodes with a thick red border represent the nodes in $H_2$ that are actually attacked, and blue empty nodes represent the rest of the nodes.}
\vspace*{-0.2cm}
\label{fig:compare_bad_case}
\end{figure*}

Fig.~\ref{fig:eval_H1} shows the performance of the REACT Algorithm in detecting the attacked area and recovering the information after data distortion and data replay attacks on the attacked area $H_1$ accompanied by single, double, and triple line failures. As can be seen in Fig.~\ref{fig:ExtraNodes_H1}, the REACT Algorithm can exactly detect the attacked area after all attack scenarios under both the distortion attack and the replay attack. Hence, the performance of the REACT Algorithm is almost the same in detecting line failures and recovering the phase angles after both data attack scenarios.

Fig.~\ref{fig:FNFP_H1} shows the average number of False Negatives (FN) and False Positives (FP) in detecting line failures. As can be seen, the REACT Algorithm can detect line failures with very low average number of FNs and FPs. Moreover, as it is shown in Fig.~\ref{fig:Exact_Rec_H1}, the REACT Algorithm exactly detects single, double, and triple line failures in 94\%, 87\%, and 82\% of the cases, respectively.


Fig.~\ref{fig:Run_Time_H1} shows the average running time of the REACT Algorithm in detecting all attacked scenarios in this case. Our system has an Intel Core i7-2600 \@3.40GHz CPU and 16GB
RAM. One can see that the running time of the REACT Algorithm is very low. The average confidence of the solutions are also shown in Fig.~\ref{fig:Conf_H1}. As can be seen, despite few false negatives and positives in detecting line failures, the solutions obtained by the REACT Algorithm have very high confidence which means that the REACT Algorithm barely missed finding the correct solution.

Finally, Fig.~\ref{fig:Phase_H1} shows the average percentage error in the recovered phase angles. It can be seen that the phase angles inside the attacked area can be recovered with less than 3\%, 5\%, and 7\% error after the single, double, and triple line failures, respectively.

\begin{figure}[t]
\vspace*{-0.2cm}
\centering
\begin{subfigure}{0.23\textwidth}
\centering
\includegraphics[scale=0.28]{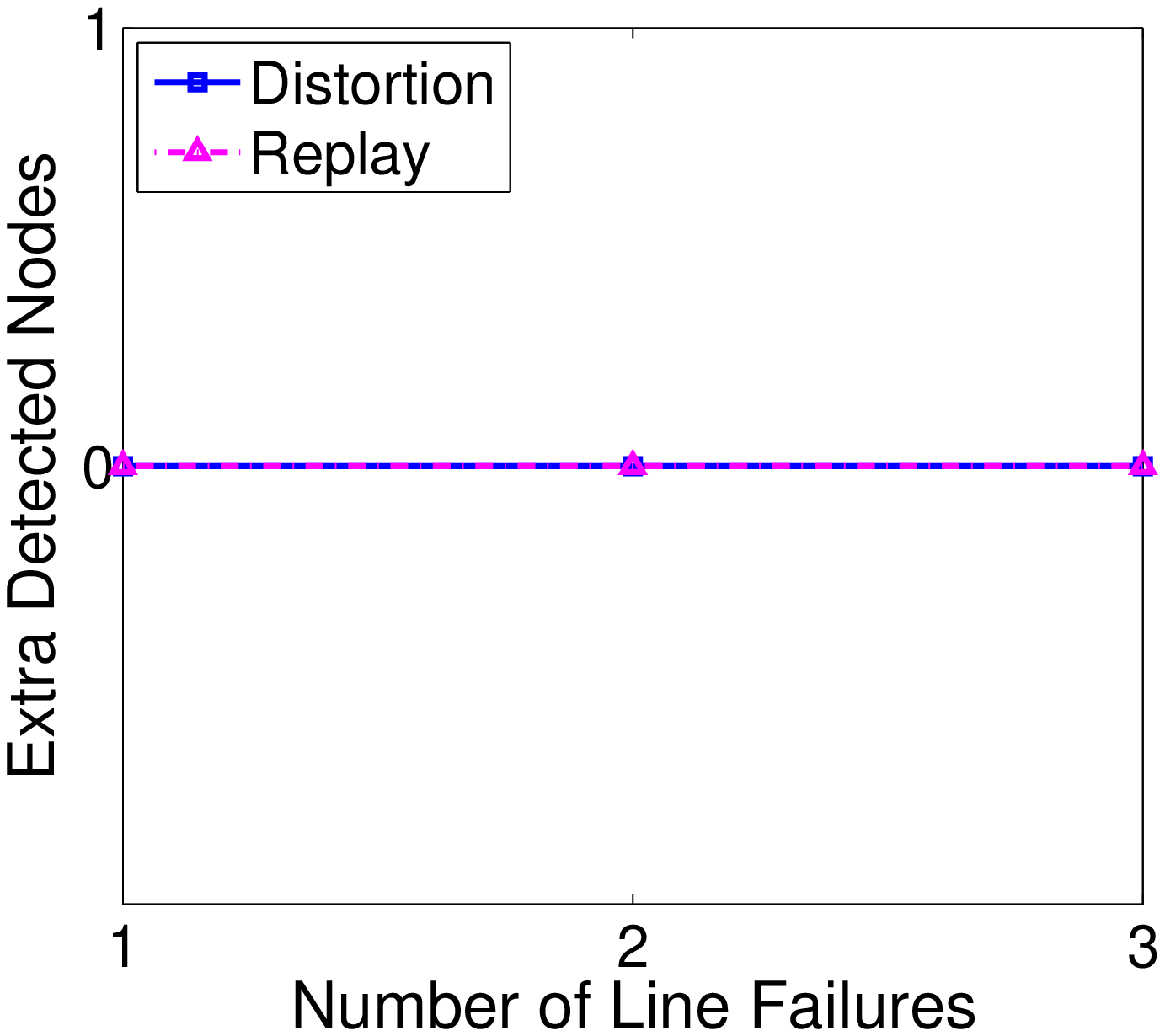}
\vspace*{-0.5cm}
\caption{}
\label{fig:ExtraNodes_H1}
\end{subfigure}
\begin{subfigure}{0.23\textwidth}
\centering
\includegraphics[scale=0.28]{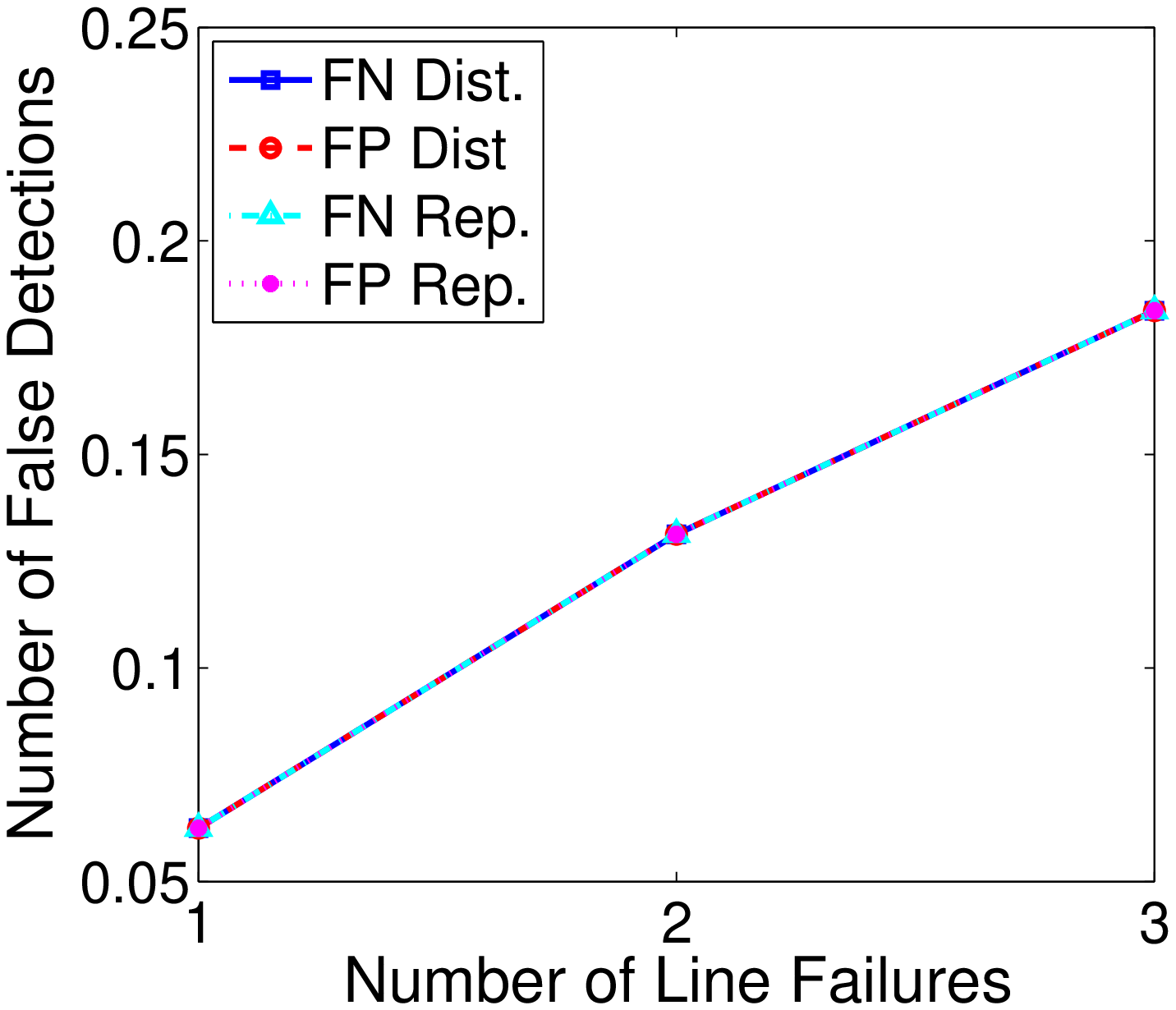}
\vspace*{-0.5cm}
\caption{}
\label{fig:FNFP_H1}
\end{subfigure}
\begin{subfigure}{0.23\textwidth}
\vspace*{-0.07cm}
\centering
\includegraphics[scale=0.28]{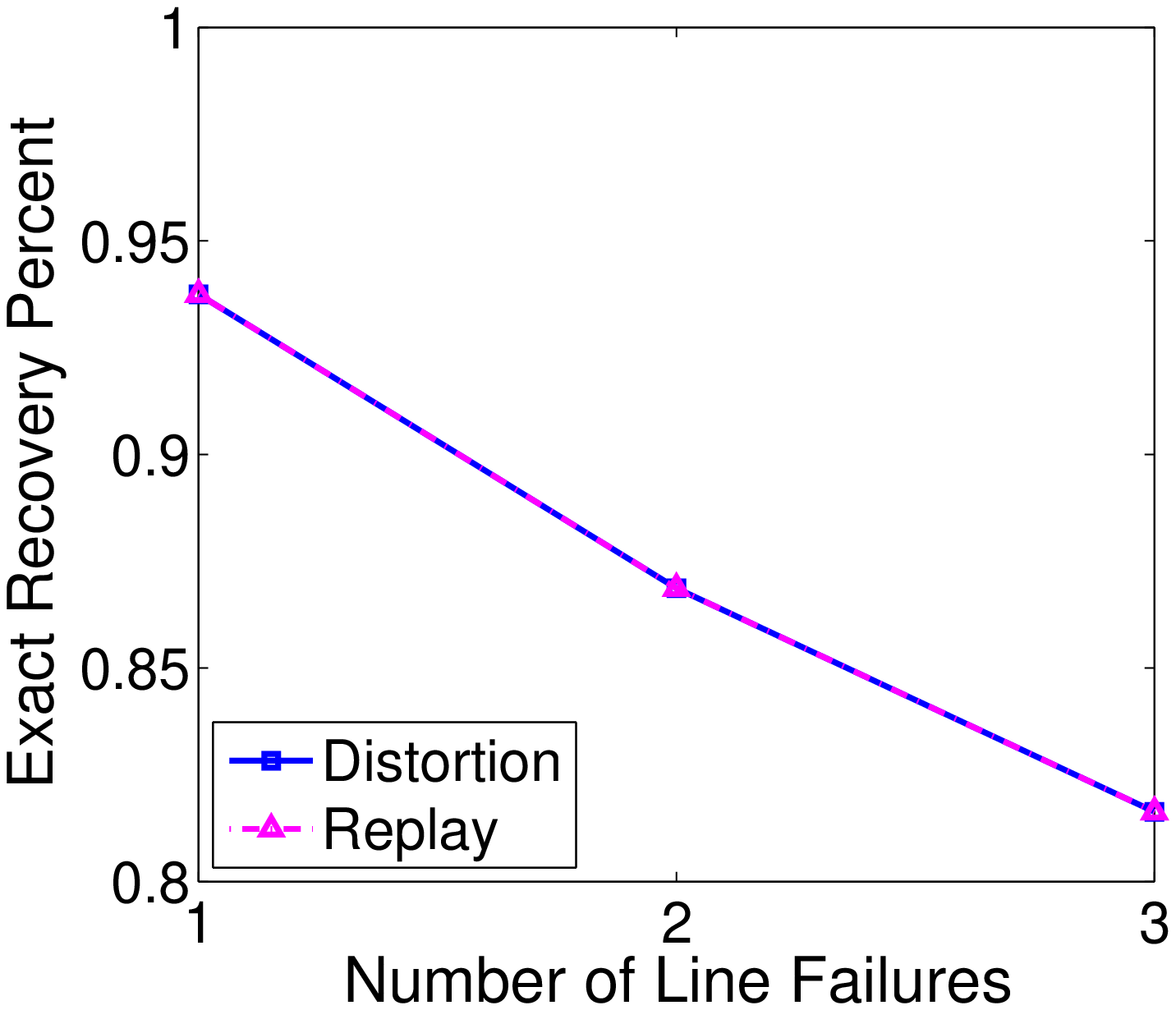}
\vspace*{-0.5cm}
\caption{}
\label{fig:Exact_Rec_H1}
\end{subfigure}
\begin{subfigure}{0.23\textwidth}
\vspace*{-0.07cm}
\centering
\includegraphics[scale=0.28]{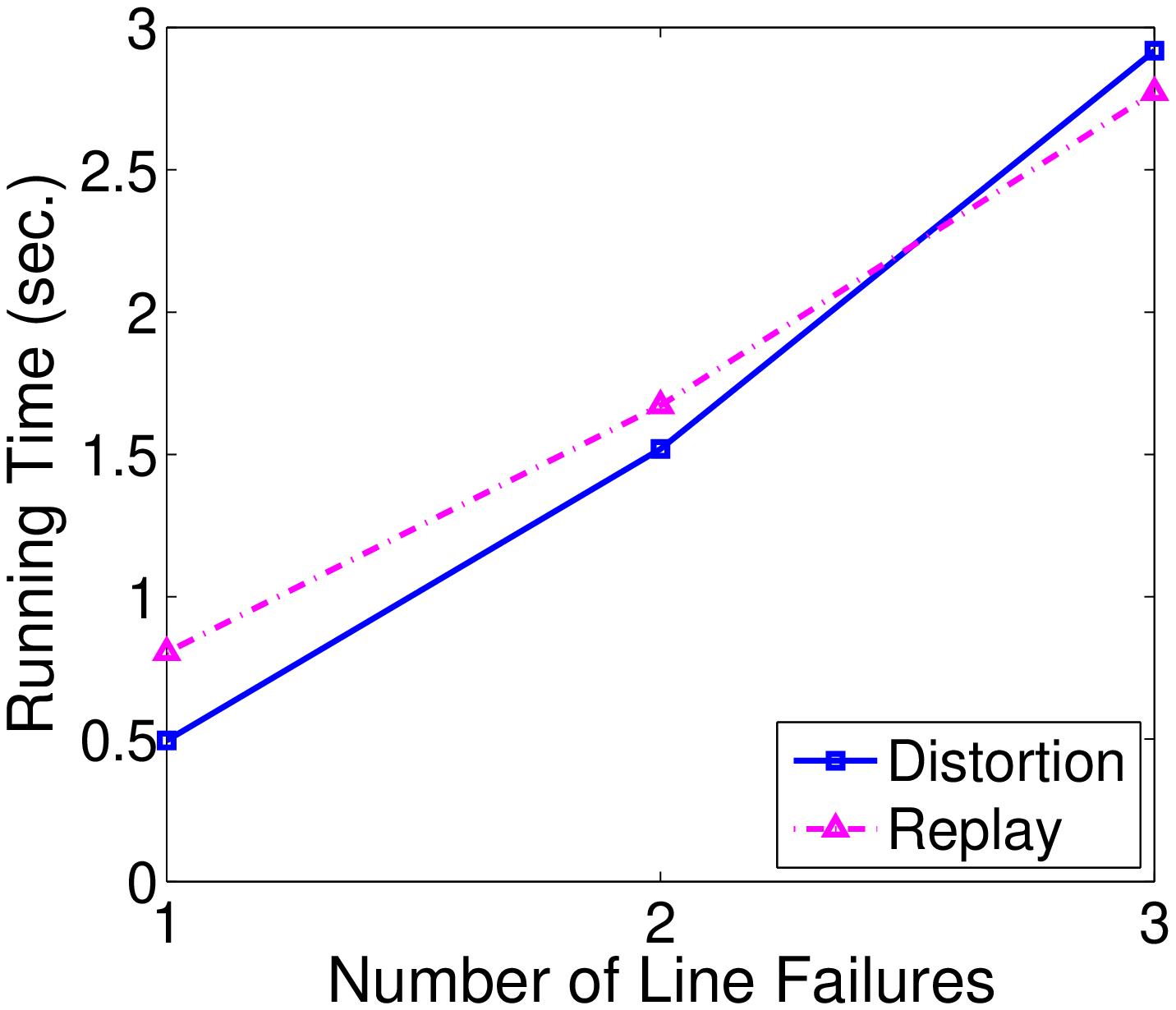}
\vspace*{-0.5cm}
\caption{}
\label{fig:Run_Time_H1}
\end{subfigure}
\begin{subfigure}{0.23\textwidth}
\vspace*{-0.07cm}
\centering
\includegraphics[scale=0.28]{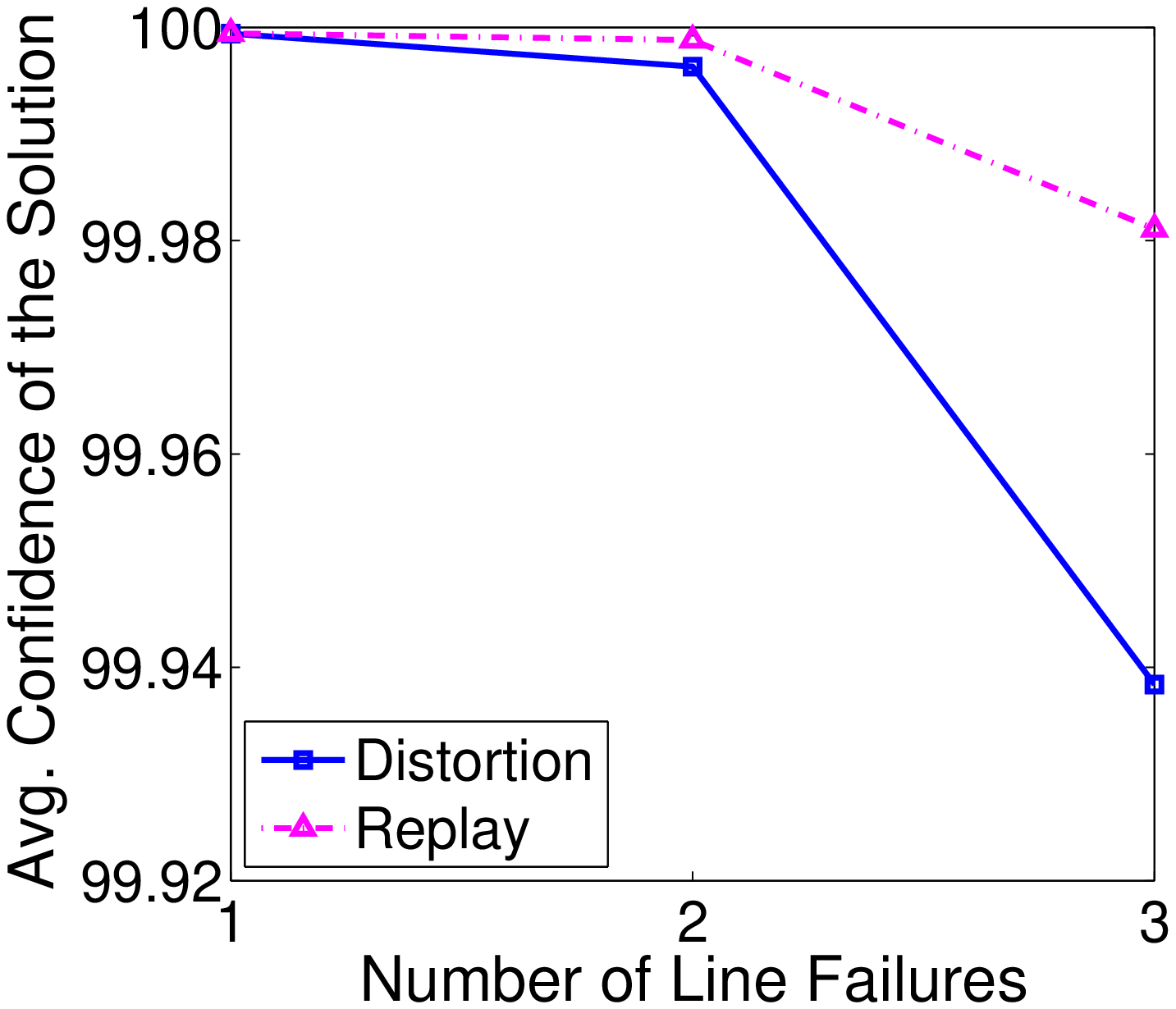}
\vspace*{-0.5cm}
\caption{}
\label{fig:Conf_H1}
\end{subfigure}
\begin{subfigure}{0.23\textwidth}
\vspace*{-0.07cm}
\centering
\includegraphics[scale=0.28]{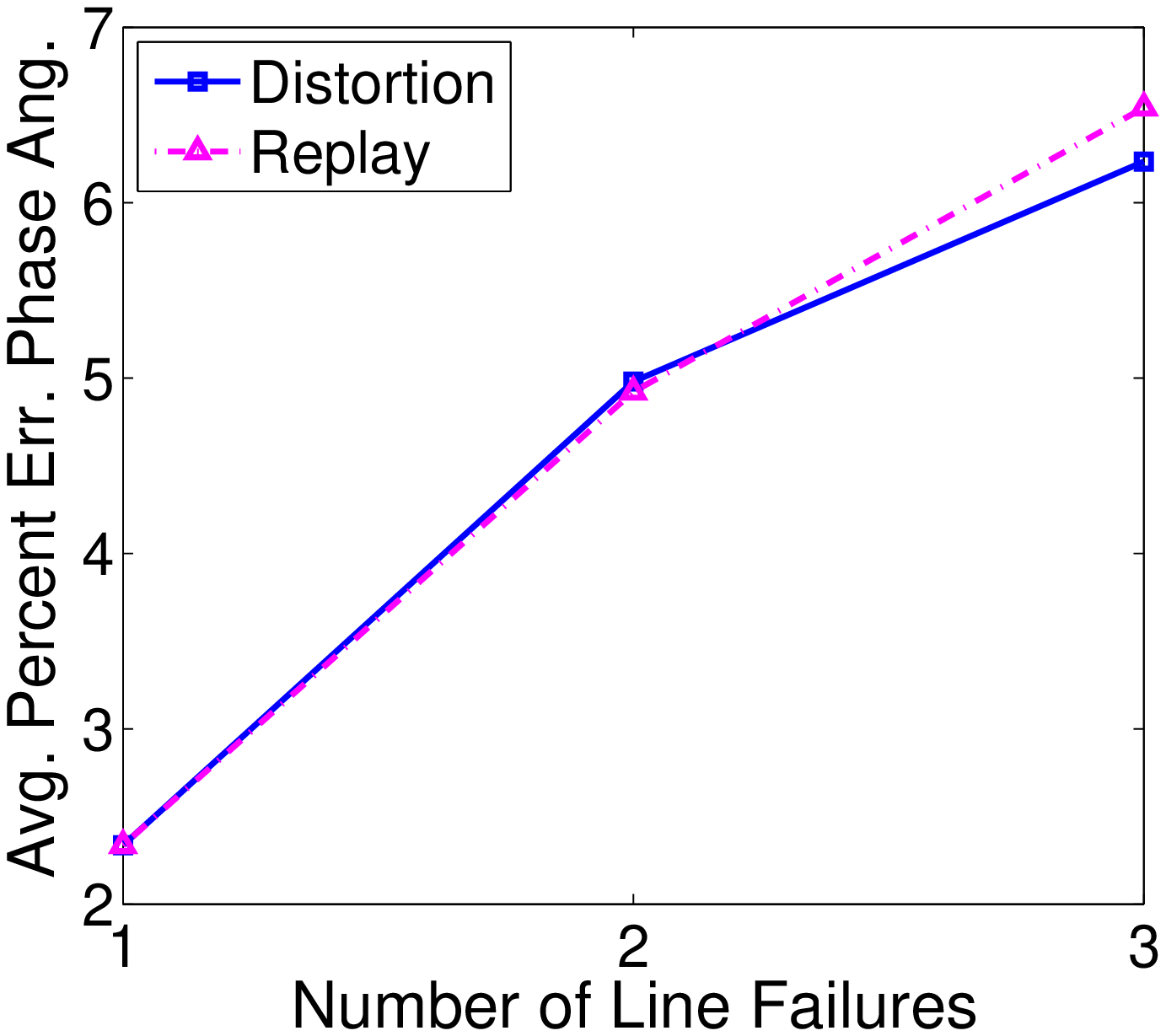}
\vspace*{-0.5cm}
\caption{}
\label{fig:Phase_H1}
\end{subfigure}
\caption{The REACT Algorithm's performance in detecting the attacked area and recovering the information after data distortion and replay attacks on the attacked area $H_1$ accompanied by single, double, and triple line failures. (a) Average number of extra nodes detected as attacked in detecting the attacked area, (b) average number of false positives and negatives in detecting line failures, (c) percentage of the cases with exact line failures detection, (d) running time of the algorithm, (e) average confidence of the solutions, and (f) average error in recovered phase angles.}
\label{fig:eval_H1}
\end{figure}

As we observed in Fig.~\ref{fig:eval_H1}, when the attacked area is relatively small, the REACT Algorithm performs very similarly after the two types of data attack. However, as it can be clearly seen in Fig.~\ref{fig:eval_H2}, it is not the case as the attacked area becomes larger. Before we analyze the results provided in Fig.~\ref{fig:eval_H2}, in order to better show the difficulty of detecting the attacked area after a data replay attack, we depicted in Fig.~\ref{fig:compare_bad_case} one of the analyzed attacked scenarios in Fig.~\ref{fig:eval_H2}. As can be seen in Fig.~\ref{fig:dist}, the REACT Algorithm can exactly detect the attacked area after a data distortion attack on $H_2$ which is accompanied by a triple line failure. However, it may have difficulties detecting the attacked area after a data replay attack on the same area with the same set of line failures. Recall from Subsection~\ref{subsec:data_replay} that the main reason for this is the difficulty of distinguishing between the nodes in $\intt(H)$ and $\intt(\bar{H})$.

\begin{figure}[t]
\vspace*{-0.2cm}
\centering
\begin{subfigure}{0.23\textwidth}
\centering
\includegraphics[scale=0.28]{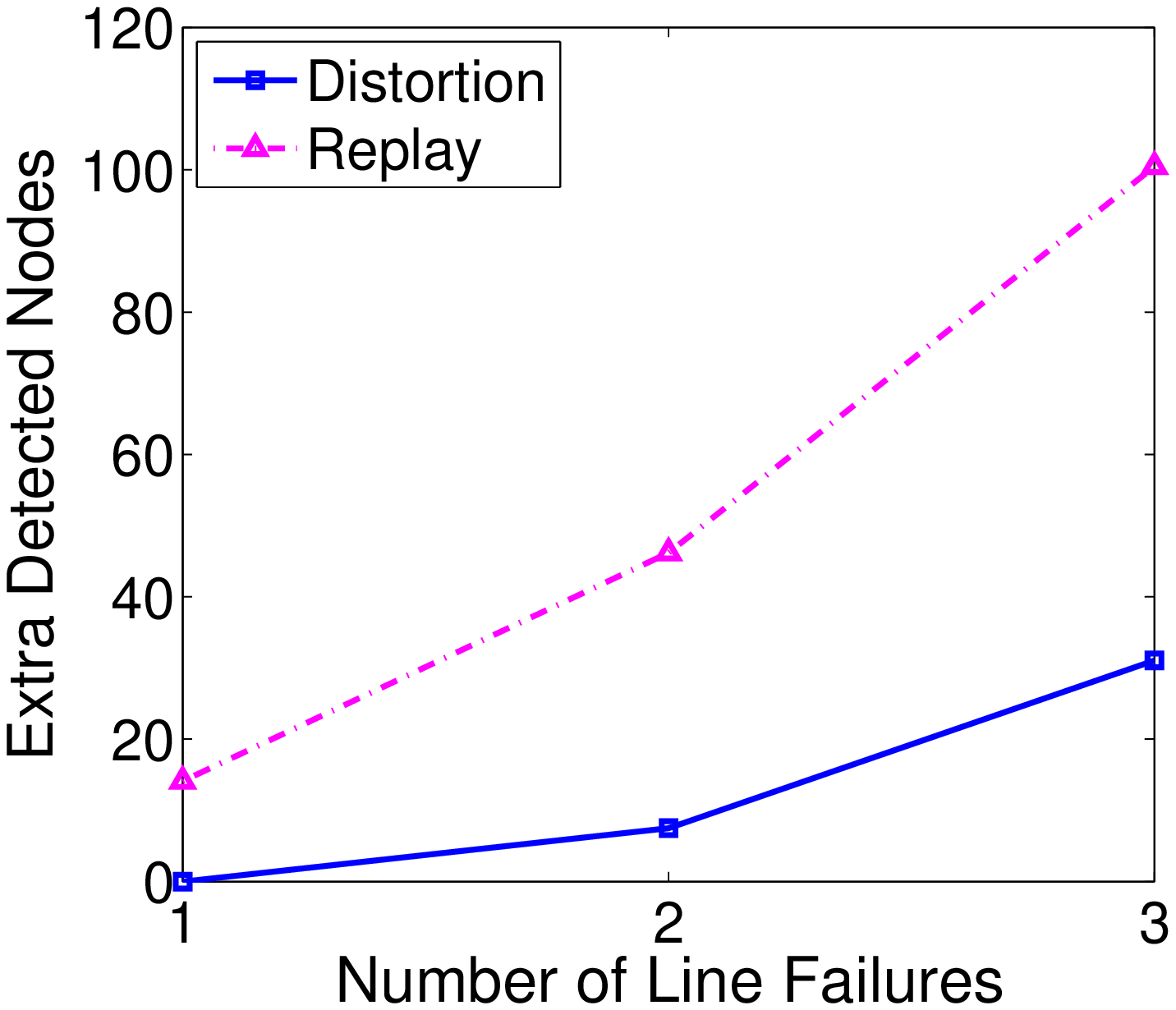}
\vspace*{-0.5cm}
\caption{}
\label{fig:ExtraNodes_H2}
\end{subfigure}
\begin{subfigure}{0.23\textwidth}
\centering
\includegraphics[scale=0.28]{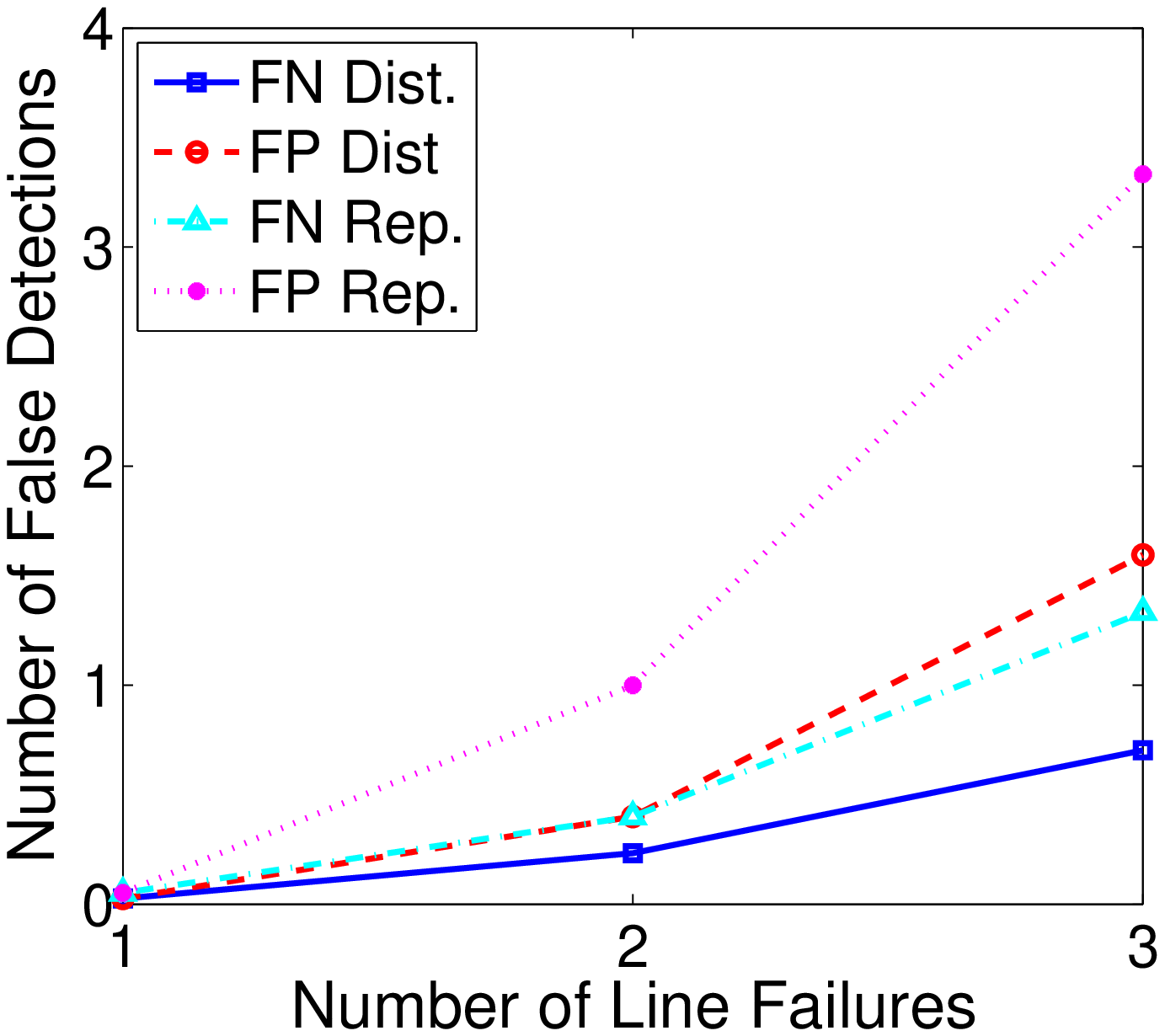}
\vspace*{-0.5cm}
\caption{}
\label{fig:FNFP_H2}
\end{subfigure}
\begin{subfigure}{0.23\textwidth}
\vspace*{-0.07cm}
\centering
\includegraphics[scale=0.28]{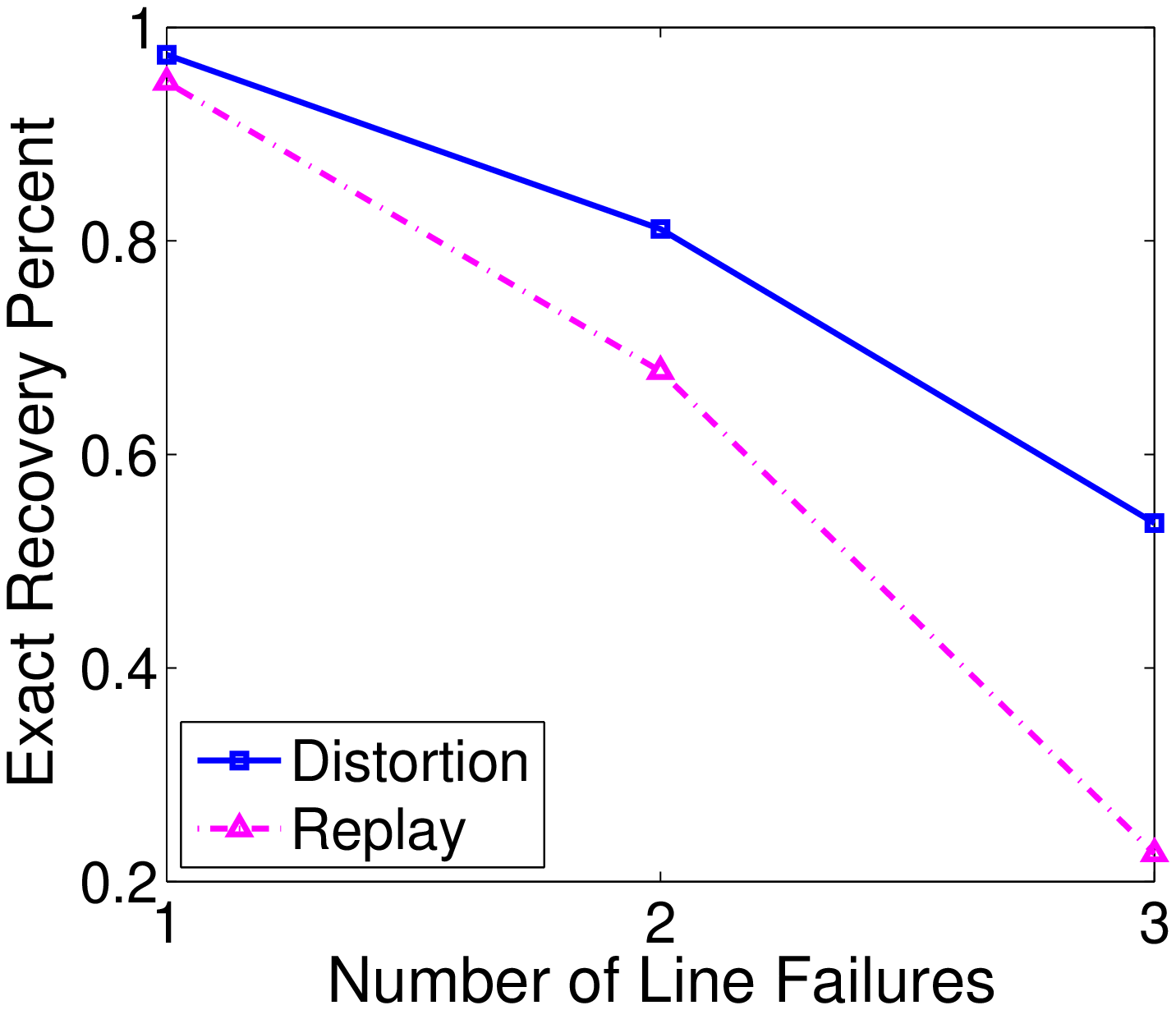}
\vspace*{-0.5cm}
\caption{}
\label{fig:Exact_Rec_H2}
\end{subfigure}
\begin{subfigure}{0.23\textwidth}
\vspace*{-0.07cm}
\centering
\includegraphics[scale=0.28]{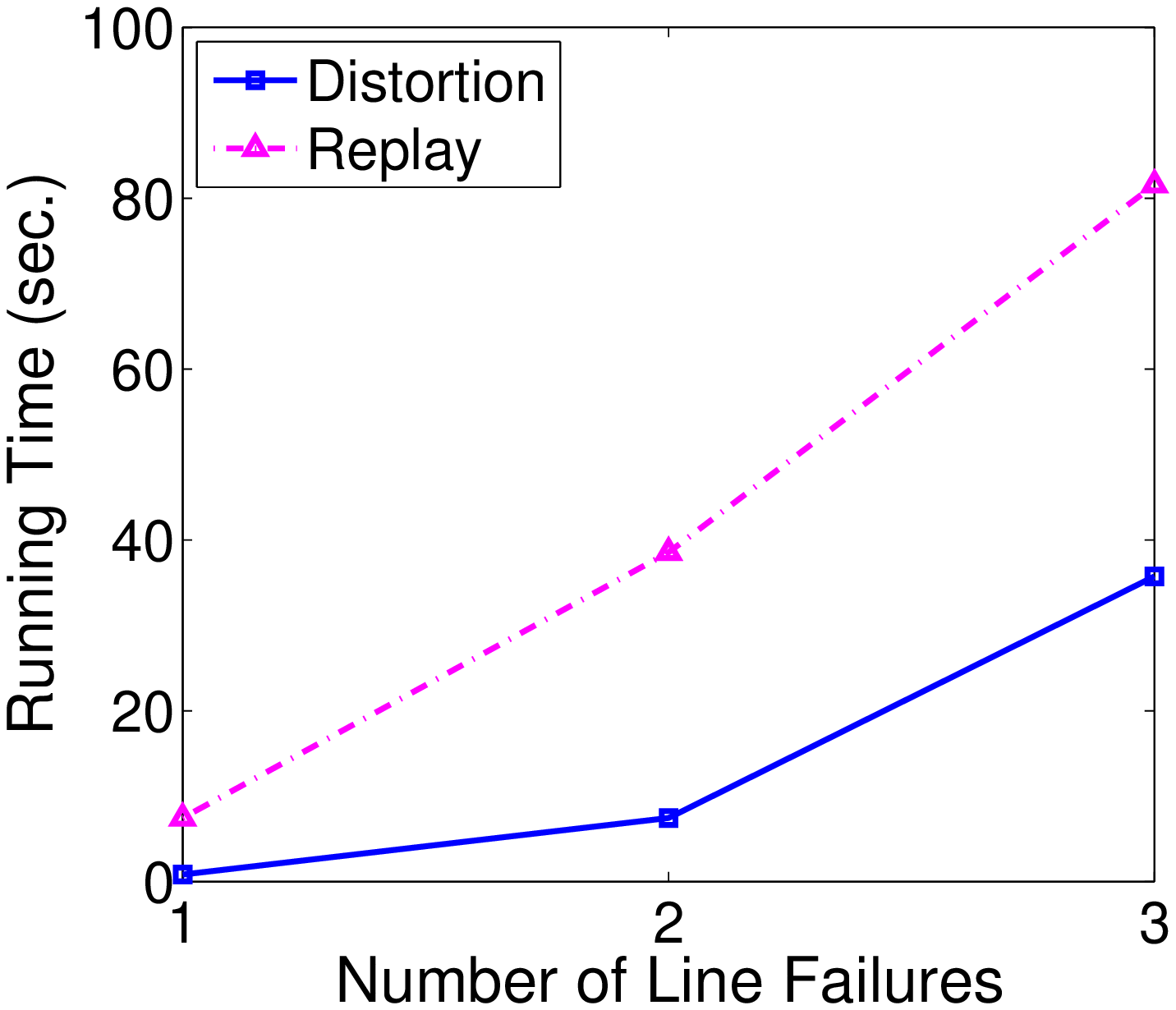}
\vspace*{-0.5cm}
\caption{}
\label{fig:Run_Time_H2}
\end{subfigure}
\begin{subfigure}{0.23\textwidth}
\vspace*{-0.07cm}
\centering
\includegraphics[scale=0.28]{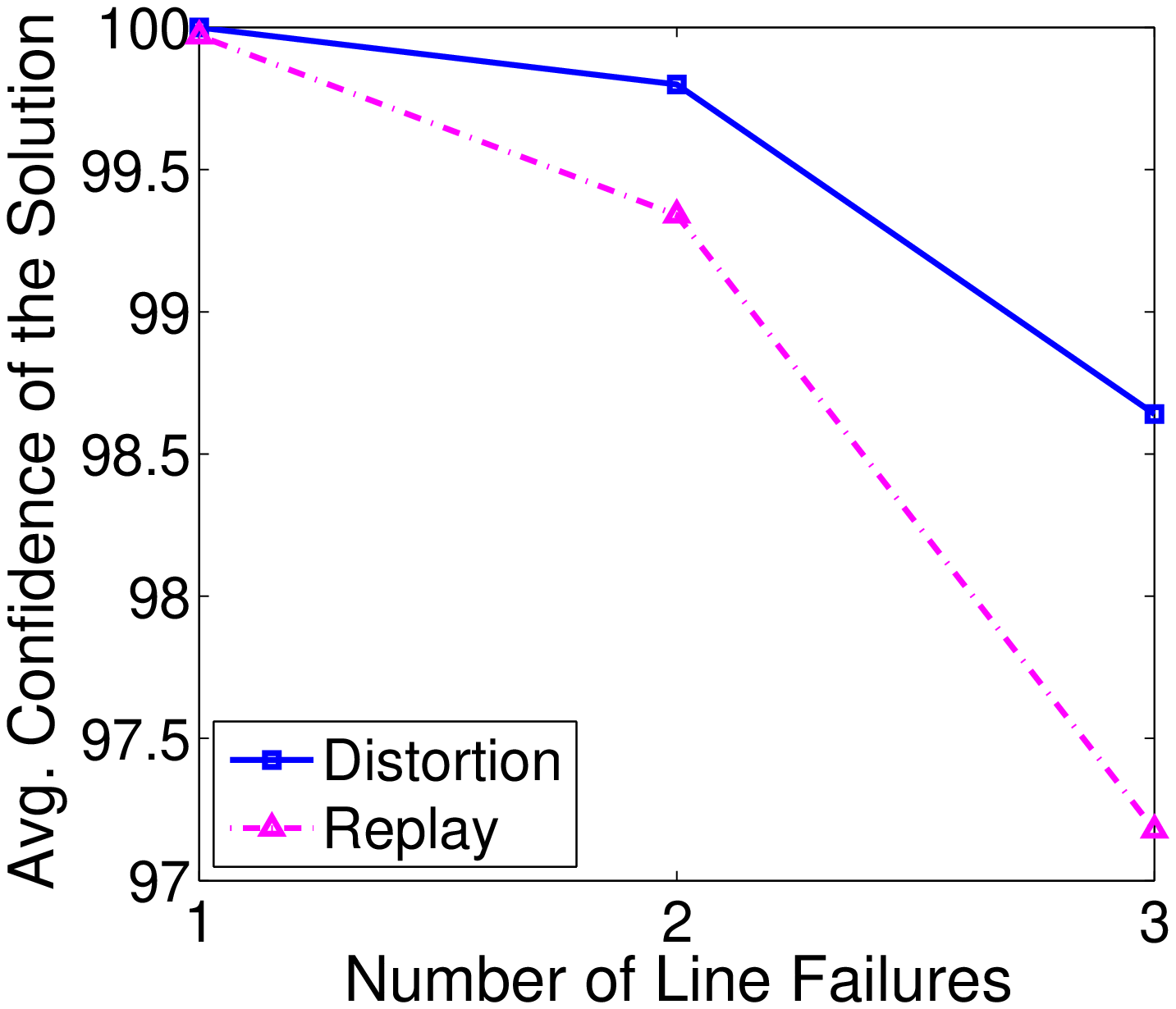}
\vspace*{-0.5cm}
\caption{}
\label{fig:Conf_H2}
\end{subfigure}
\begin{subfigure}{0.23\textwidth}
\vspace*{-0.07cm}
\centering
\includegraphics[scale=0.28]{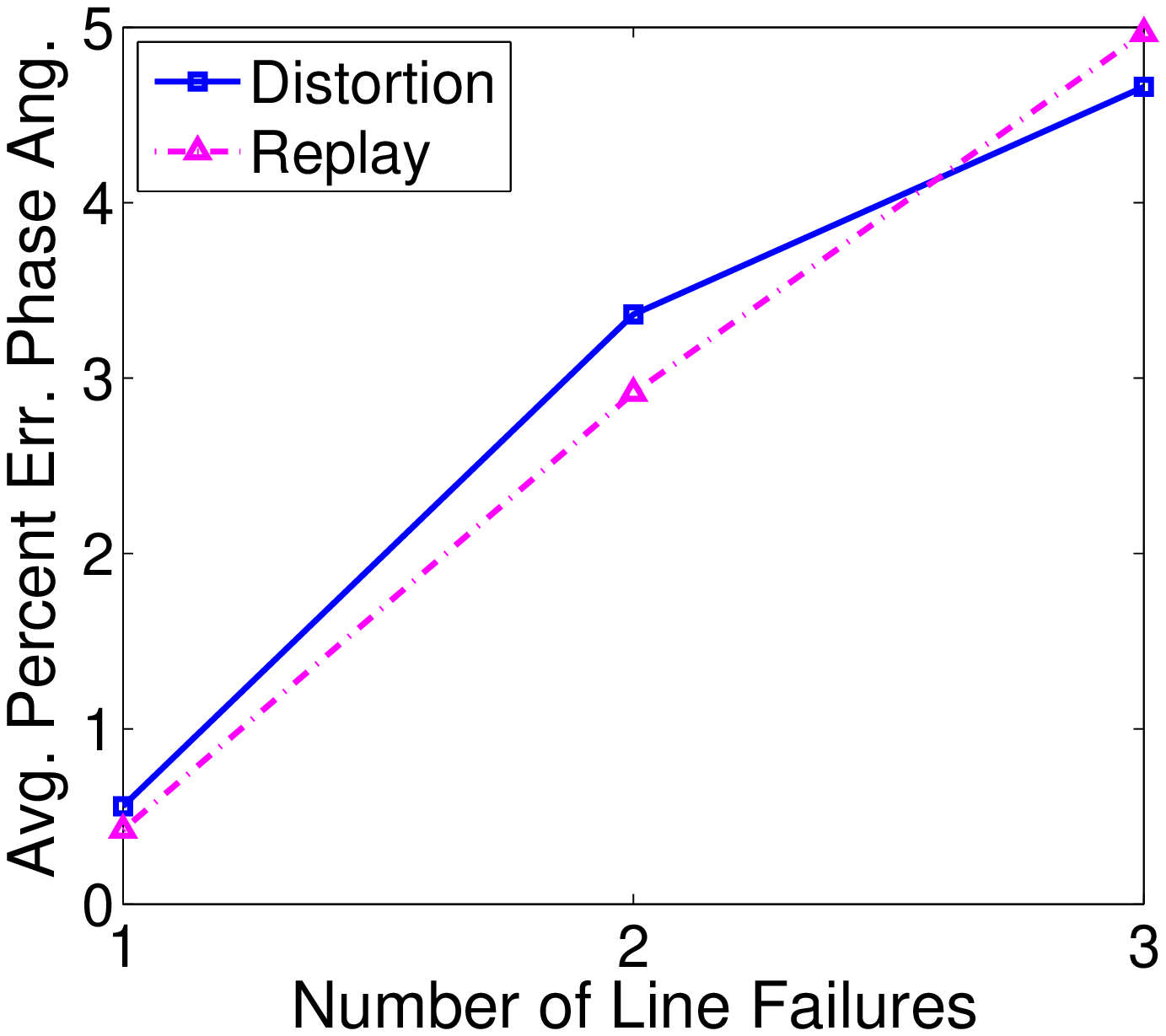}
\vspace*{-0.5cm}
\caption{}
\label{fig:Phase_H2}
\end{subfigure}
\caption{The REACT Algorithm's performance in detecting the attacked area and recovering the information after data distortion and replay attacks on the attacked area $H_2$ accompanied by single, double, and triple line failures. (a) Average number of extra nodes detected as attacked in detecting the attacked area, (b) average number of false positives and negatives in detecting line failures, (c) percentage of the cases with exact line failures detection, (d) running time of the algorithm, (e) average confidence of the solutions, and (f) average  error in recovered phase angles.}
\label{fig:eval_H2}
\end{figure}

Fig.~\ref{fig:ExtraNodes_H2} shows the extra nodes that are incorrectly detected by the REACT Algorithm as part of the attacked area. As can be seen, in the case of the data distortion attack, the number of line failures do not significantly affect the performance of the REACT Algorithm. However, in the case of the replay attack, as the number of line failures within the attacked area increases, the REACT Algorithm provides less accurate approximation of the attacked area.

Despite its difficulty in detecting the attacked area after a data replay attack, Figs.~\ref{fig:FNFP_H2} and \ref{fig:Exact_Rec_H2} demonstrate that the REACT Algorithm detects the line failures relatively accurately. For example, the REACT Algorithm accurately detects the single and double line failures in 95\% and 65\% of the cases, respectively.

As can be seen in Fig.~\ref{fig:Run_Time_H2}, the running time of the REACT Algorithm increases as the size of the attacked area increases. However, it still detects line failures much faster than existing brute force methods~\cite{tate2008line,tate2009double,zhu2012sparse,zhao2012pmu,zhu2014phasor} (to the best of our knowledge there are no methods for detecting the attacked area).

Similar to the previous attack scenario, one can see in Fig.~\ref{fig:Conf_H2} that the confidence of the solutions obtained by the REACT Algorithm are very high. It means that in these attack scenarios, many good solutions exist near the optimal solution.  This demonstrates another difficulty of dealing with recovery of information after a cyber attack on the power grid.

Finally, Fig.~\ref{fig:Phase_H2} indicates that the REACT Algorithm performs very well in recovering the phase angles in this case as well. As can be seen, for both the data distortion and the data replay attacks accompanied by single, double, and triple line failures, the REACT Algorithm recovers the phase angles with less than 5\% error.

Overall, the simulation results in this section demonstrate that the REACT Algorithm performs very well in detecting the attacked area and the line failures when the attacked area is relatively small. As the attacked area becomes larger, the Algorithm still performs very well in detecting the attacked area after a distortion data attack. However, it may face difficulties providing an accurate approximation of the attacked area after a replay attack. Despite this, in both data attack scenarios, it detects line failures relatively well. One of the important observations in this section is that the LIFD module outperforms the methods provided in~\cite{SYZ2015} for detecting line failures with an slight increase in the running time, since it needs to find a solution to (\ref{eq:weighted_simul_detect}) several times instead of once. The results in this section clearly demonstrate that in the attacked areas $H_1$ and $H_2$ that do not have the conditions provided~\cite{SYZ2015}, the LIFD Module can still detect the line failures relatively accurately with less than 20 iterations. In most of theses cases, the LIFD Module detects the line failures within much fewer number of iterations. 
\section{Conclusion}\label{sec:conclusion}
In this paper, we considered a model for cyber attacks on power grids focusing on both data distortion and data reply attacks. We proved that the problem of detecting the line failures after such an attack is NP-hard in general and even when the attacked area is known. However, using the algebraic properties of the DC power flows, we developed the polynomial time REACT Algorithm for approximating the attacked area and detecting the line failures after a cyber attack on the grid. We numerically showed that the REACT Algorithm obtains accurate results when there are few number of line failures and the attacked area is small. We showed that as the attacked area becomes larger and the number of line failures increases, the REACT Algorithm faces some difficulties but still can approximate the attacked area and detect line failures with few false negatives and positives.

We analytically and numerically showed that the data replay attacks are harder to deal with than the data distortion attacks. It is possible for an adversary to devise more sophisticated attacks to further obscure the system's state. Studying more sophisticated attacks and improving the methods to protect the grid against such attacks is part of our future work. Moreover, since the DC power flows only provide an approximation for the more accurate AC power flows, we plan to extend our methods to function under the AC power flow model as well.



\normalsize
\section{Appendix: Omitted Proofs}\label{sec:proofs}
\begin{proof}[Proof of Lemma~\ref{lem:expdistbound}]
Define $s_k : = \sum_{i=1}^k w_i$. It is known that
\begin{equation*}
f_{s_k}(x) = \frac{\lambda e^{-\lambda x} (\lambda x)^{k-1}}{(k-1)!}.
\end{equation*}
Now since $w_i$s are i.i.d.\ random variables, $\sum_{i=k+1}^m w_i\sim s_{m-k}$. Therefore, all we need to compute is $Pr(s_k<s_{m-k})$.
\begin{align}
&Pr(\sum_{i=1}^k w_i<\sum_{i=k+1}^m w_i)= \int_0^\infty\! Pr(s_{m-k}-s_k=a)~da\nonumber\\
&~= \int_0^\infty\!\!\!\int_0^\infty\! Pr(s_k=y) Pr(s_{m-k}=y+a)~ dy~da\nonumber\\
&~= \int_0^\infty\!\!\!\int_0^\infty\! \frac{\lambda e^{-\lambda y} (\lambda y)^{k-1}}{(k-1)!} \frac{\lambda e^{-\lambda(y+a)}(\lambda(y+a))^{m-k-1}}{(m-k-1)!}~ dy~da\nonumber\\
&~= \int_0^\infty\!\! \frac{\lambda^m e^{-2\lambda y}y^{k-1}}{(k-1)!(m-k-1)!}\Big(\!\int_0^\infty\!\!\! e^{-\lambda a} (y+a)^{(m-k-1)}da\Big)dy.\label{eq:proofexp1}
\end{align}
On the other hand, by defining $z:=\lambda (y+a)$, we have:
\begin{align*}
\int_0^\infty\!\!\! e^{-\lambda a} (y+a)^{(m-k-1)}da=\frac{e^\lambda y}{\lambda ^{m-k}}\int_{\lambda y}^{\infty} e^{-z} z^{m-k-1}~dz.\\
\end{align*}
Define $T(n+1):=\int_{\lambda y}^\infty e^{-z} z^n~dz$. Using partial integration:
\begin{align*}
T(n+1) &= \left[-e^{-z} z^n\right]_{\lambda y}^\infty + \int_{\lambda y}^\infty n z^{n-1} e^{-z}~dz\\
&= e^{-\lambda y} (\lambda y)^n +  n T(n)= n!e^{-\lambda y}\sum_{i=0}^{n} \frac{(\lambda y)^{i}}{i!}.
\end{align*}
Using equation above in (\ref{eq:proofexp1}) results in:
\begin{align*}
Pr(\sum_{i=1}^k w_i\!<\!\!\!\!\sum_{i=k+1}^m w_i)&= \int_0^\infty\!\! \frac{\lambda^n e^{-2\lambda y}y^{k-1}}{(k-1)!(m-k-1)!} \frac{e^{\lambda y}}{\lambda^{m-k}}T(m-k)\nonumber\\
&= \frac{\lambda^k}{(k-1)!}\int_0^\infty e^{-2\lambda y} y^{k-1}\Big(\sum_{i=0}^{m-k-1}\frac{(\lambda y)^i}{i!}\Big)dy\nonumber\\
&= \frac{\lambda^k}{(k-1)!}\sum_{i=0}^{m-k-1}\Big(\int_0^\infty e^{-2\lambda y} y^{k-1}\frac{(\lambda y)^i}{i!}~dy\Big)\nonumber.\\
\end{align*}
By defining $x:=2\lambda y$ and using Gamma function:
\begin{align}
Pr(\sum_{i=1}^k w_i\!<\!\!\!\!\sum_{i=k+1}^m w_i)&= \frac{\lambda^k}{(k-1)!}\sum_{i=0}^{m-k-1}\Big(\frac{\lambda^{-k}}{i!2^{i+k}}\int_0^\infty e^{-x} x^{k+i-1}~dx\Big)\nonumber\\
&= \frac{\lambda^k}{(k-1)!}\sum_{i=0}^{m-k-1}\Big(\frac{\lambda^{-k}}{i!2^{i+k}}(k+i-1)!\Big)\nonumber\\
&=\sum_{i=0}^{m-k-1}2^{-i-k} \binom{k+i-1}{i}\nonumber\\
&=2^{-(m-1)}\sum_{i=0}^{m-k-1}2^{(m-1)-(i+k)} \binom{k+i-1}{k-1}.\label{eq:proofexp2}
\end{align}
Now notice that $\sum_{i=0}^{m-k-1}2^{(m-1)-(i+k)} \binom{k+i-1}{i}$ is equal to the total number of subsets of $\{1,\dots,m-1\}$ with at least $k$ elements. The reason is that this summation is equal to the total number of subsets that contain $k+i$ and exactly $k-1$ elements from $\{1,2,\dots,k+i-1\}$. It is easy to verify that by summing this up on $i$, we count all the subsets of $\{1,\dots,m-1\}$ with at least $k$ elements. On the other hand, we can count the total number of subsets of $\{1,\dots,m-1\}$ with at least $k$ elements using the complement rule. The total number of subsets with at least $k$ elements is equal to the total number of subsets minus number of subsets of size 0,1,\dots,$k-1$. Hence,
\begin{equation*}
\sum_{i=0}^{m-k-1}2^{(m-1)-(i+k)} \binom{k+i-1}{k-1} = 2^{m-1}-\sum_{i=0}^{k-1} \binom{m-1}{j}.
\end{equation*}
 Now using the equation above in (\ref{eq:proofexp2}) and using the equality $2^{m-1} = \sum_{i=0}^{m-1} \binom{m-1}{j}$, proves the lemma.
\end{proof}
\begin{proof}[Proof of Corollary~\ref{cor:expdistbound}]
It is easy to see that if $k\leq(m-1)/2$, then $\sum_{j=k}^{m-1}\binom{m-1}{j}\geq 2^{m-2}$. Therefore from Lemma~\ref{lem:expdistbound}, $Pr(\sum_{i=1}^k w_i<\sum_{i=k+1}^m w_i)\geq 1/2$ and there is nothing left to prove. So assume $k=m/2+\Theta(\sqrt{m})$.
It is proved in~\cite[Lemma 10.8]{macwilliams1977theory} that for any $1/2<\alpha<1$,
\begin{equation*}
\frac{2^{n\textbf{H}(\alpha)}}{\sqrt{8n\alpha(1-\alpha)}}\leq \sum_{j=\alpha n}^{n}\binom{n}{k},
\end{equation*}
in which $\textbf{H}(\alpha)=-\alpha \log_2(\alpha)-(1-\alpha)\log_2(1-\alpha)$ is the entropy function. Now to prove Corollary~\ref{cor:expdistbound}, select $n=m-1$, and $\alpha = 1/2+\epsilon$ for $\epsilon = \Theta(1/\sqrt{n})$. First notice that one can show that the Taylor expansion of the entropy function around $1/2$ can be computed as:
\begin{equation*}
\textbf{H}(\alpha) = 1-\frac{1}{2\ln 2}\sum_{i=1}^{\infty} \frac{(1-2\alpha)^{2i}}{i(2i-1)}.
\end{equation*}
Using approximation above, it is easy to see that $\textbf{H}(\alpha)\approx 1-\Theta(\epsilon^2)=1-\Theta(1/n)$. Hence, $2^{n\textbf{H}(\alpha)}=2^{n-\Theta(1)}$. On the other hand,
\begin{align*}
\sqrt{8n\alpha(1-\alpha)} &= \sqrt{8n(1/2+\epsilon)(1/2-\epsilon)} = \sqrt{8n(1/4-\epsilon^2)}\\
& = \sqrt{2n-\Theta(1)}\approx \Theta(\sqrt{n}).
\end{align*}
Hence, by replacing $n$ by $m-1$ and using Lemma~\ref{lem:expdistbound}, one can verify:
\begin{equation*}
Pr(\sum_{i=1}^k w_i<\sum_{i=k+1}^m w_i)= \Omega(\frac{1}{\sqrt{m}}).
\end{equation*}
\end{proof} 

\section*{Acknowledgement}
This work was supported in part by DARPA RADICS under
contract \#FA-8750-16-C-0054, NSF grants CCF-1423100 and CCF-1703925, funding from the U.S. DOE
OE as part of the DOE Grid Modernization Initiative, and DTRA grant HDTRA1-13-1-0021.
\bibliographystyle{abbrv}
\bibliography{bib_REACT}
\end{document}